\documentclass[journal, onecolumn]{IEEEtran}

\usepackage{graphicx}
\usepackage{amsfonts}
\usepackage{amsmath}
\usepackage{amssymb}
\usepackage{color}
\usepackage{bm} 

\newtheorem{thm}{Theorem}
\newtheorem{cor}[thm]{Corollary}
\newtheorem{lem}[thm]{Lemma}
\newtheorem{proof}[thm]{proof}

\newtheorem{defn}[thm]{Definition}
\newtheorem{rem}[thm]{Remark}
\newtheorem{exam}[thm]{Example}

\begin{document}

\title{On Sequential Locally Repairable Codes}

\author{Wentu Song, ~ Kai Cai ~ and ~ Chau Yuen
\thanks{W. Song and C. Yuen are with Singapore University of
Technology and Design, Singapore
       (e-mails: \{wentu\_song, yuenchau\}@sutd.edu.sg).}
\thanks{Kai Cai is with the Department of Mathematics, University of Hong
Kong. (e-mail: kcai@hku.hk).} }

\maketitle

\begin{abstract}
We consider the locally repairable codes (LRC), aiming at
sequential recovering multiple erasures. We define the $(n, k, r,
t)$-SLRC (Sequential Locally Repairable Codes) as an $[n,k]$
linear code where any $t'(\leq t)$ erasures can be sequentially
recovered, each one by $r~(2\leq r<k)$ other code symbols.
Sequential recovering means that the erased symbols are recovered
one by one, and {\em an already recovered symbol can be used} to
recover the remaining erased symbols. This important recovering
method, in contrast with the vastly studied parallel recovering,
is currently far from understanding, say, lacking codes
constructed for arbitrary $t\geq 3$ erasures and bounds to
evaluate the performance of such codes.

We first derive a tight upper bound on the code rate of $(n, k, r,
t)$-SLRC for $t=3$ and $r\geq 2$. We then propose two
constructions of binary $(n, k, r, t)$-SLRCs for general $r,t\geq
2$ (Existing constructions are dealing with $t\leq 7$ erasures).
The first construction generalizes the method of direct product
construction. The second construction is based on the resolvable
configurations and yields SLRCs for any $r\geq 2$ and odd $t\geq
3$. For both constructions, the rates are optimal for
$t\in\{2,3\}$ and are higher than most of the existing LRC
families for arbitrary $t\geq 4$.
\end{abstract}


\IEEEpeerreviewmaketitle \setcounter{equation}{0}

\section{Introduction}
To avoid the inefficiency of straightforward replication of data,
various coding techniques are introduced to the distributed
storage system (DSS), among which the linear locally repairable
codes, also known as locally recoverable codes (LRC)
\cite{Gopalan12}, \cite{Oggier11}, attracted much attention
recently. Roughly speaking, a linear LRC with locality $r$ is an
$[n,k]$ linear code such that the value of each coordinate (code
symbol) can be computed from the values of at most $r$ other
coordinates.

In a DSS system where a LRC $\mathcal C$ is used, the information
stored in each storage node corresponds to one coordinate of
$\mathcal C$. Hence, each {\em single} node failure (erasure) can
be recovered by a set of at most $r$ other nodes. However, it is
very common that two or more storage nodes fail in the system.
This problem, which has become a central focus for the LRC
society, are recently investigated by many authors $($e.g.
\cite{Pamies13}$-$\cite{Balaji16-2}$)$. Basically, when multiple
erasures occur, the recovering performance can be heavily depends
on the recovering strategy in use, say, recovering the erasures
simultaneously or one by one. The two strategies were first
distinguished as \emph{parallel approach} and \emph{sequential
approach} in \cite{Prakash-14}. Comparing with the parallel
approach, the sequential approach recovery erasures one by one and
hence the already fixed erasure nodes can be used in the next
round of recovering. Potentially, for the same LRC, using the
sequential approach can fix more erasures than using the parallel
approach, and hence the sequential approach is a better candidate
than the parallel approach in practice. However, due to technique
difficulties, this more important approach remains far from
understood, say, lacking of both code constructions and bounds to
evaluate the code performance. In contrast with the vastly studied
parallel approach \cite{Pamies13}$-$\cite{Wang16}, existing work
on the sequential approach up to date are limited to dealing with
$t\leq7$ erasures.
For example, the case of $t=2$ are
considered in \cite{Prakash-14}, where the authors derived upper
bounds on the code rate as well as minimum distance and also
constructed a family of distance-optimal codes based on Tur\'{a}n
graphs. For the code rate, they proved that:
\begin{align}\label{rate-bd-3}
\frac{k}{n}\leq \frac{r}{r+2}.
\end{align}
The original version of this work \cite{Wentu-15}, firstly
considered the case of $t=3$ and gave both constructions and code
rate bounds for $t\in\{2,3\}~($in a more generalized manner of
functional recovering$)$. Of great relevance to the present work
are the results recently obtained in \cite{Balaji16} and
\cite{Balaji16-2}, where the authors derived a lower bound on code
length $n$ of \emph{binary code} for $t=3$ and an upper bound on
code rate of \emph{binary code} for $t=4$. A couple of optimal or
high rate constructions were provided in these two papers, say,
rate-optimal codes for $t\in\{2,3,4\}$, and high rate codes for
$r=2$ and $t\in\{5,6,7\}$. Here, we note that, by using orthogonal
Latin squares, the authors in \cite{Balaji16} gave an interesting
construction of sequential locally recoverable codes for {\em any
odd $t\geq 3$} with rate
$k/n=1/\left(1+\frac{t-1}{r}+\frac{1}{r^2} \right)$. Obviously,
the SLRCs can deal with any $t$ erasures and having high code rate
and are highly desired in both theory and practices.

\subsection{Our Contribution}
In practice, high rate LRCs are desired since they mean low
storage overhead. In this work, we are interested in the high rate
LRCs for sequential recovering any $t\geq 3$ erasures, by defining
the $(n,k,r,t)$-SLRC (Sequential Locally Repairable Code) as an
$[n, k]$ linear code in which any $t'~(t'\leq t)$ erased code
symbols can be sequentially recovered, {\em each one} by at most
$r~(2\leq r<k)$ other symbols. Our first contribution is an upper
bound on the code rate for $(n,k,r,t)$-SLRC with $t=3$ and any
$k>r\geq 2$. The bound is derived by using a graph theoretical
method, say, we associate each $(n,k,r,t)$-SLRC with a set of
directed acyclic graphs, called repair graphs, and then obtain the
bound by studying the structural properties of the so-called {\em
minimal repair graph}. The sprit of this method lies in
\cite{Fragouli06, Wentu11}. For general $t\geq 5$, deriving an
achievable, explicit upper bound of the rate of $(n,k,r,t)$-SLRC
seems very challenging, and we give some discussions and
conjectures on this issue.

Then we construct two families of binary $(n,k,r,t)$-SLRC. The
first family, which contains the product of $m$ copies of the
binary $[r+1, r]$ single-parity code \cite{Tamo14} as a special
case, is for any positive integers $r~(\geq 2)$ and $t$, and has
rate
$$\frac{k}{n}=\frac{1}{\sum_{s=0}^t\frac{1}{r^{|\text{supp}_m(s)|}}},$$
where $m$ is any given positive integer satisfying $t\leq 2^m-1$
and $\text{supp}_m(s)$ is the support of the $m$-digit binary
representation\footnote{The $m$-digit binary representation of any
positive integer $s\leq 2^m-1$ is the binary vector
$(\lambda_m,\lambda_{m-1},\cdots,\lambda_1)\in\mathbb Z_2^m$ such
that $s=\sum_{j=1}^m\lambda_j2^{j-1}$.} of $s$. The second family
is constructed for any $r\geq 2$ and any odd integer $t\geq 3$ and
is based on {\em resolvable configurations}. This family has code
rate
$$\frac{k}{n}=\left(1+\frac{t-1}{r}+\left\lceil\frac{1}{r^2}
\right\rceil\right)^{-1},$$ which is the same with the Latin
square-based code constructed in \cite{Balaji16}. For
$t\in\{2,3\}$, the code rates of these two constructions are
optimal.

A basic and important fact revealed by our study is: the
sequential approach can have much better performance than parallel
approach, e.g, for the direct product of $m$ copies of the binary
$[r+1,r]$ single-parity code, it can recover $m$ erasures with
locality $r$ by the parallel approach \cite{Tamo14}, but $2^m-1$
erasures with the same locality by the sequential approach.

\subsection{Related Work}
Except that mentioned previously, most existing work focus on
$[n,k]$ linear LRCs with parallel approach. In \cite{Prakash12},
the authors defined and constructed the {\em $(r,t+1)_a$ code},
for which each code symbol $i$ is contained in a punctured code
(local code) with length $\leq r+t$ and minimum distance $\geq
t+1$. Clearly, for such codes, any $t$ erased code symbols can be
recovered in parallel by at most $tr$ other code symbols, among
which, each erased symbol can be recovered by at most $r$ symbols.
The code rate of this family satisfies \cite{Wentu14}
\begin{align}\label{rate-bd-1}\frac{k}{n}\leq
\frac{r}{r+t}.\end{align} Another family is the codes with
locality $r$ and availability $t$ \cite{Wang14,Rawat14}, for
which, each code symbol has $t$ disjoint recovering sets of size
at most $r$. An upper bound on the code rate of such codes is
proved in \cite{Tamo14}:
\begin{align}\label{rate-bd-2}\frac{k}{n}\leq
\frac{1}{\prod_{j=1}^{t}(1+\frac{1}{jr})}.\end{align}
Unfortunately, for $t\geq 3$, the tightness of bound
\eqref{rate-bd-2} is not known and most of existing construction
have rate $\leq\frac{r}{r+t}~($e.g., see
\cite{Wang15,Huang15,Wang16}$)$. Constructions with rate
$>\frac{r}{r+t}$ are proposed only for some very special values,
e.g., $(n,k,r,t)=(2^{r+1}-1,2^{r}-1,r,r+1)$ \cite{Wang16}. The
third family of parallel recovery LRC is proposed in
\cite{Pamies13}, in which, for any set $E\subseteq[n]$ of erasures
of size at most $t$ and any $i\in E$, the $i$th code symbol has a
recovering set of size at most $r$ contained in $[n]\backslash E$.
The fourth family, called codes with cooperative local repair, is
proposed in \cite{Rawat-14} and defined by a stronger condition:
each subset of $t$ code symbols can be {\em cooperatively}
recovered from at most $r$ other code symbols. For this family, an
upper bound of the code rate with exactly the same form as
\eqref{rate-bd-1} is derived \cite{Rawat-14}. By far, constructing
LRCs with high code rate $($e.g., $\frac{k}{n}>\frac{r}{r+t})$ is
still an interesting open problem, both for parallel recovery and
for sequential recovery.

\subsection{Organization}
The rest of this paper is organized as follows. In Section
\uppercase\expandafter{\romannumeral 2}, we define the
$(n,k,r,t)$-SLRC and then present some basic and useful facts. In
section \uppercase\expandafter{\romannumeral 3}, we first
investigate the (minimal) repair graphs of the SLRC and then prove
the upper bound on the code rate of $(n,k,r,t)$-SLRC for
$t\in\{2,3\}$. Before constructing the first family of SLRC in
Section \uppercase\expandafter{\romannumeral 5}, we first study an
example in Section \uppercase\expandafter{\romannumeral 4}. Then,
the second family of SLRC is constructed in Section
\uppercase\expandafter{\romannumeral 6}. Finally, the paper is
concluded in Section \uppercase\expandafter{\romannumeral 7}.

\subsection{Notations}
For any positive integer $n$, $[n]:=\{1,2,\cdots,n\}$. For any set
$A$, $|A|$ is the size (the number of elements) of $A$. If
$B\subseteq A$ and $|B|=t$, then $B$ is called a $t$-subset of
$A$. For any real number $x$, $\lceil x\rceil$ is the smallest
integer greater than or equal to $x$. If $\mathcal C$ is an
$[n,k]$ linear code and $A\subseteq[n]$, then $\mathcal C|_{A}$
denotes the punctured code by puncturing coordinates in
$\overline{A}:=[n]\backslash A$. For any codeword
$x=(x_1,x_2,\cdots,x_n)\in\mathcal C$, $\text{supp}(x):=\{i\in[n];
x_i\neq 0\}$ is the support of $x$.

\section{Preliminary}

\subsection{Sequential Locally repairable code (SLRC)}
Let $\mathcal C$ be an $[n,k]$ linear code over the finite field
$\mathbb F$ and $i\in[n]$. A subset $R\subseteq[n]\backslash\{i\}$
is called a \emph{recovering set} of $i$ if there exists an
$a_j\in\mathbb F\backslash\{0\}$ for each $j\in R$ such that
$x_i=\sum_{j\in R}a_jx_j$ for all
$x=(x_1,x_2,\cdots,x_n)\in\mathcal C$. Equivalently, there exists
a codeword $y$ in the dual code $C^\bot$ such that
$\text{supp}(y)=R\cup\{i\}$.

\begin{defn}[Sequential Locally Repairable Code]\label{e-lrc}
For any $E\subseteq[n]$, $\mathcal C$ is said to be
$(E,r)$-recoverable if $E$ can be sequentially indexed, say
$E=\{i_1,i_2,\cdots,i_{|E|}\}$, such that each $i_\ell\in E$ has a
recovering set
$R_{\ell}\subseteq\overline{E}\cup\{i_1,\cdots,i_{\ell-1}\}$ of
size $|R_\ell|\leq r$, where $\overline{E}:=[n]\backslash E$;
$\mathcal C$ is called an $(n,k,r,t)$-\emph{sequential locally
repairable code (SLRC)} (or simply $(r,t)$-SLRC) if $\mathcal C$
is $(E,r)$-recoverable for each $E\subseteq[n]$ of size $|E|\leq
t$, where $r$ is called the locality of $\mathcal C$.
\end{defn}

As a special case of Definition \ref{e-lrc}, if for each
$E\subseteq[n]$ of size $|E|\leq t$ and each $i\in E$, $i$ has a
recovering set $R\subseteq\overline{E}$ of size $|R|\leq r$, then
$\mathcal C$ is called an $(n,k,r,t)$-\emph{parallel locally
repairable code (PLRC)}. This special case is first considered in
\cite{Pamies13}.

By the definition, we can have $r\leq k$ for any $(n,k,r,t)$-SLRC.
{\em Throughout this paper, we assume that a recovering set $R$
has size  $2\leq |R|\leq r<k$}. The following equivalent form of
Definition \ref{e-lrc} will be frequently used in our paper.
\begin{lem}\label{lem-ELRC}
$\mathcal C$ is an $(n,k,r,t)$-SLRC if and only if for any
nonempty $E\subseteq[n]$ of size $|E|\leq t$, there exists an
$i\in E$ such that $i$ has a recovering set
$R\subseteq[n]\backslash E$.
\end{lem}
\begin{proof}
Let $\mathcal C$ be an $(n,k,r,t)$-SLRC and $\emptyset\neq
E\subseteq[n]$ of size $|E|\leq t$. Then by Definition
\ref{e-lrc}, $E$ can be sequentially indexed as
$E=\{i_1,i_2,\cdots,i_{|E|}\}$ such that $i_1$ has a recovering
set $R_1\subseteq[n]\backslash E$.

Conversely, for any $E\subseteq[n]$ of size $|E|\leq t$, by
assumption, one can find an $i_1\in E$ such that $i_1$ has a
recovering set $R_1\subseteq[n]\backslash E$. Further, since
$|E\backslash\{i_1\}|<|E|\leq t$, then by assumption, there exists
an $i_2\in E\backslash\{i_1\}$ such that $i_2$ has a recovering
set $R_2\subseteq[n]\backslash
\left(E\backslash\{i_1\}\right)=\overline{E}\cup\{i_1\}$.
Similarly, we can find an $i_3\in E\backslash\{i_1,i_2\}$ such
that $i_3$ has a recovering set
$R_3\subseteq\overline{E}\cup\{i_1,i_2\}$, and so on. Then $E$ can
be sequentially indexed as $E=\{i_1,i_2,\cdots,i_{|E|}\}$ such
that each $i_\ell\in E$ has a recovering set
$R_\ell\subseteq\overline{E}\cup\{i_1,\cdots,i_{\ell-1}\}$. So by
definition \ref{e-lrc}, $\mathcal C$ is an $(n,k,r,t)$-SLRC.
\end{proof}

The following lemma gives a sufficient condition of $(r,t)$-SLRC,
which reflects the difference between the sequential recovery and the
parallel recovery.
\begin{lem}\label{lem-Com-LRC}
Suppose $[n]=A\cup B$ and $A\cap B=\emptyset$. Suppose
$t_1,t_2\geq 0$ and $\mathcal C$ is an $[n,k]$ linear code such
that
\begin{itemize}
 \item [(1)] For any nonempty $E\subseteq A$ of size $|E|\leq t_1$,
 there exists an $i\in E$ such that $i$ has a recovering
 set $R\subseteq A\backslash E$;
 \item [(2)] For any nonempty $E\subseteq A$ of size $|E|\leq t_1+t_2+1$,
 there exists an $i\in E$ such that $i$ has a recovering
 set $R\subseteq[n]\backslash E$;
 \item [(3)] For any nonempty $E\subseteq B$ of size $|E|\leq t_2$,
 there exists an $i\in E$ such that $i$ has a recovering
 set $R\subseteq B\backslash E$;
 \item [(4)] For any nonempty $E\subseteq B$ of size $|E|\leq t_1+t_2+1$,
 there exists an $i\in E$
 such that $i$ has a recovering set $R\subseteq[n]\backslash E$.
\end{itemize}
Then $\mathcal C$ is an $(r,t)$-SLRC with $t=t_1+t_2+1$.
\end{lem}
\begin{proof}
We prove, by Lemma \ref{lem-ELRC}, that for any nonempty
$E\subseteq[n]$ of size $|E|\leq t_1+t_2+1$, there exists an $i\in
E$ such that $i$ has a recovering set $R\subseteq[n]\backslash E$.
Obviously, it holds when $E\subseteq A$ or $E\subseteq B~($by
condition (2) or (4)$)$. So we assume $E\cap A\neq\emptyset$ and
$E\cap B\neq\emptyset$. Consider the following two cases.

Case 1: $0<|E\cap A|\leq t_1$. By condition (1), there exists an
$i\in E$ such that $i$ has a recovering set $R\subseteq
A\backslash E\subseteq[n]\backslash E$.

Case 2: $|E\cap A|>t_1$. Since $|E|\leq t_1+t_2+1$ and $A\cap
B=\emptyset$, then $0<|E\cap B|\leq t_2$. By condition (3), there
exists an $i\in E$ such that $i$ has a recovering set $R\subseteq
B\backslash E\subseteq[n]\backslash E$.

The proof is completed by combining the above cases.
\end{proof}

\subsection{Repair Graph and Minimal Repair Graph}
Let $G=(\mathcal V,\mathcal E)$ be a directed, acyclic graph,
where $\mathcal V$ is the vertex set and $\mathcal E$ is the
(directed) edge set. A directed edge $e$ from vertex $u$ to $v$ is
denoted by an ordered pair $e=(u,v)$, where $u$ is called the
\emph{tail} of $e$ and $v$ the \emph{head} of $e$. Moreover, $u$
is called an \emph{in-neighbor} of $v$ and $v$ an
\emph{out-neighbor} of $u$. For each $v\in \mathcal V$, let
$\text{In}(v)$ and $\text{Out}(v)$ denote the set of in-neighbors
and out-neighbors of $v$ respectively. If
$\text{In}(v)=\emptyset$, we call $v$ a \emph{source}; otherwise,
$v$ is called an \emph{inner vertex}. Denote by $\text{S}(G)$ the
set of all sources of $G$. For any $E\subseteq\mathcal V$, let
\begin{align}\label{Def-Out-E}
\text{Out}(E)=\bigcup_{v\in E}\text{Out}(v)\backslash E.
\end{align} By \eqref{Def-Out-E}, we have $E\cap\text{Out}(E)=\emptyset$.
For any $v\in\mathcal V$, denote
\begin{align}\label{Def-Out-2}
\text{Out}^2(v)=\bigcup_{u\in\text{Out}(v)}\text{Out}(u)
\backslash\text{Out}(v)
\end{align}
i.e., $\text{Out}^2(v)$ is the set of all $w\in \mathcal V$ such
that $w$ is an out-neighbor of some $u\in\text{Out}(v)$ but not an
out-neighbor of $v$.

As an example, consider the graph depicted in Fig. \ref{eg-fg-1},
where vertices are indexed by $\{1,2,\cdots, 16\}$. Then
$\text{Out}(3)=\{9,10\}$, $\text{Out}(4)=\{10,11\}$ and
$\text{Out}^2(3)=\{13,15,16\}$. Let $E=\{3,4\}$. Then
$\text{Out}(E)=\{9,10,11\}$.

\renewcommand\figurename{Fig}
\begin{figure}[htbp]
\begin{center}
\includegraphics[height=3.6cm]{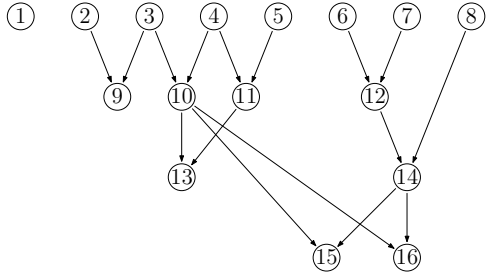}
\end{center}
\vspace{-0.2cm}\caption{An example repair graph with $n=16$, $r=2$
and $|S(G)|=8$. }\label{eg-fg-1}
\end{figure}

\begin{defn}[Repair Graph]\label{def-rg}
Let $\mathcal C$ be an $(n,k,r,t)$-SLRC and $G=(\mathcal
V,\mathcal E)$ be a directed, {\em acyclic} graph such that
$\mathcal V=[n]$. $G$ is called a {\em repair graph} of $\mathcal
C$ if for all inner vertex $i\in\mathcal V$, $\text{In}(i)$ is a
recovering set of $i$.
\end{defn}

Obviously, an $(n,k,r,t)$-SLRC may have many repair graphs. If
$\mathcal C$ is an $(n,k,r,t)$-SLRC, we usually use $\{G_\lambda;
\lambda\in\Lambda\}$ to denote the set of all repair graphs of
$\mathcal C$, where $\Lambda$ is some proper index set. It should
be noted that the repair graph defined here has subtle differences
with the recovering graph defined in \cite{Tamo14}, e.g., it must
be acyclic and an $(n,k,r,t)$-SLRC may have many repair graphs
such that for each $i\in[n]$, at most one recovery set of $i$ is
considered in each repair graph. The key ingredient of our
technique is the so-called {\em minimal repair graph} as defined
follows.

Let $\mathcal C$ be an $(n,k,r,t)$-SLRC and $\{G_\lambda;
\lambda\in\Lambda\}$ be the set of all repair graphs of $\mathcal
C$. Recall that for each $\lambda\in\Lambda$, $\text{S}(G_\lambda)$
is the set of all sources of $G_\lambda$. Denote
\begin{align}\label{dlt-star} \delta^*\triangleq
\min\{|\text{S}(G_\lambda)|;\lambda\in\Lambda\}.
\end{align}

\begin{defn}[Minimal Repair Graph]\label{def-mrg}
A repair graph $G_{\lambda_0}$, $\lambda_0\in\Lambda$, is called a
\emph{minimal repair graph} of $\mathcal C$ if
$|\text{S}(G_{\lambda_0})|=\delta^*$.
\end{defn}

\begin{rem}\label{ext-min-gph}
It is easy to see that any $(n,k,r,t)$-SLRC has at least one
minimal repair graph by noticing that the set
$\{|\text{S}(G_\lambda)|;\lambda\in\Lambda\}\subseteq [n]$ is
finite.
\end{rem}

\section{An Upper Bound on the Code Rate}
Before proposing the main result of this section, we need first
investigate properties of the minimal repair graphs of
$(n,k,r,t)$-SLRC.

\subsection{Properties of the Minimal Repair Graph}
In this subsection, we always assume that $\mathcal C$ is an
$(n,k,r,t)$-SLRC and $G_{\lambda_0}=(\mathcal V,\mathcal E)$ is a
minimal repair graph of $\mathcal C$. The following two results
are of fundamental.
\begin{lem}\label{edge-num-left}
\begin{align}
(n-\delta^*)r\geq |\mathcal E|.
\end{align}
\end{lem}
\begin{proof}
By the definition, $G_{\lambda_0}$ has $n-\delta^*$ inner vertices
and each of them has at most $r$ in-neighbors, and hence the
result follows.
\end{proof}
\begin{lem}\label{dim-dlt}
\begin{align}
\label{k-leq-delta} k\leq\delta^*.
\end{align}
\end{lem}
\begin{proof}
According to Definition \ref{def-rg}, for each $j\in[n]$, the
$j$th code symbol of $\mathcal C$ is a linear combination of the
code symbols in $\text{In}(j)$. In other words, the code symbols
of $\text{In}(j)$ spans the code symbols of
$\{j\}\cup\text{In}(j)$. Moreover, since $G_{\lambda_0}$ is
acyclic, then inductively, the code symbols of
$\text{S}(G_{\lambda_0})$ spans $\mathcal C$, which proves
$k\leq|\text{S}(G_{\lambda_0})|=\delta^*$.
\end{proof}

The following is a key lemma to investigate the structure
of $G_{\lambda_0}$.
\begin{lem}\label{mrg-out}
For any $E\subseteq[n]$ of size $|E|\leq t$,
\begin{align}\label{Out-geq-E}
|\text{Out}(E)|\geq |E\cap \text{S}(G_{\lambda_0})|.
\end{align}
\end{lem}
\begin{proof}
Suppose, on the contrary, there exists an
$E=\{i_1,i_2,\cdots,i_{t'}\}\subseteq[n]$ such that $|E|=t'\leq t$
and $|\text{Out}(E)|<|E\cap \text{S}(G_{\lambda_0})|$. By
definition of $(n, k, r, t)$-SLRC, we can let
$R_{\ell}\subseteq\overline{E}\cup\{i_1,\cdots,i_{\ell-1}\}$ be a
recovering set of $i_\ell$ for each $\ell\in[t']$.

We can construct a graph $G_{\lambda_1}$ from $G_{\lambda_0}$ by
deleting and adding edges as follows: First, for each $i\in E\cup
\text{Out}(E)$ and $j\in\text{In}(i)$, delete $(j,i)$ if it is an
edge of $G_{\lambda_0}$, and denote the resulted graph as
$G_{\lambda_1'}$; Second, for each $i_\ell\in E$ and each $j\in
R_\ell$, add a directed edge from $j$ to $i_\ell$, and let the
resulted graph be $G_{\lambda_1}$. Clearly, $G_{\lambda_1'}$ is
acyclic because $G_{\lambda_0}$ is acyclic. Moreover, since
$R_{\ell}\subseteq\overline{E}\cup\{i_1,\cdots,i_{\ell-1}\}$ for
each $\ell\in[t']$, then by construction, $G_{\lambda_1}$ is also
acyclic.

We declare that $G_{\lambda_1}$ is a repair graph of $\mathcal C$
and $|\text{S}(G_{\lambda_1})|<|\text{S}(G_{\lambda_0})|$, which
contradicts to the minimality of $G_{\lambda_0}$.

In fact, by construction,
$\text{S}(G_{\lambda_1})=(\text{S}(G_{\lambda_0})\backslash E)\cup
\text{Out}(E)$. Then for each inner node $i$ of $G_{\lambda_1}$,
we have the following two cases:

  Case 1: $i\in E$. Then $i=i_\ell$ for some $\ell\in[t']$ and by
  the construction of $G_{\lambda_1}$, $\text{In}(i)=R_\ell$ is a
  recovering set of $i$.

  Case 2: $i$ is an inner vertex of $G_{\lambda_0}$ and
  $i\notin\text{Out}(E)$. Then considering $G_{\lambda_0}$,
  $\text{In}(i)\subseteq\overline{E}=[n]\backslash E$ is a
  recovering set of $i$.

So $\text{In}(i)$ is always a recovering set of $i$. Hence,
$G_{\lambda_1}$ is a repair graph of $\mathcal C$.

On the other hand, note that by definition,
$\text{S}(G_{\lambda_0})\cap\text{Out}(E)=\emptyset$ and
$E\cap\text{Out}(E)=\emptyset$. So if we assume that
$|\text{Out}(E)|<|E\cap \text{S}(G_{\lambda_0})|$, then we have
\begin{align}\label{eq1-mrg-out}
|\text{S}(G_{\lambda_1})|&=|(\text{S}(G_{\lambda_0})\backslash
E)\cup\text{Out}(E)|\nonumber\\&=|(\text{S}(G_{\lambda_0})\backslash
E)|+|\text{Out}(E)|\nonumber\\&=|(\text{S}(G_{\lambda_0})|-|E\cap
\text{S}(G_{\lambda_0})|+|\text{Out}(E)|\nonumber\\&<|\text{S}(G_{\lambda_0})|,
\end{align}
which completes the proof.
\end{proof}

The following example illustrates the construction of
$G_{\lambda_1}$ in the proof of Lemma \ref{mrg-out}.
\begin{exam}
Consider the graph in Fig. \ref{eg-fg-1}, which we denote as
$G_{\lambda_0}$ here. Suppose it is a repair graph of a $(r=2,
t=3)$-SLRC. We can see that $\{2,3\}$ is a recovering set of $9$,
$\{3,4\}$ is a recovering set of $10$, and etc.

Let $E=\{2,3,9\}$ and assume the recovering sets of $2,3$ and $9$
are $\{1,10\}, \{12,13\}$ and $\{11,14\}$, respectively. Then we
can construct a graph $G_{\lambda_1}$ as follows. Since
$\text{Out}(E)=\{10\}$, thus, in the first step, we delete edges
$(2,9), (3,9), (3,10)$ and $(4,10)$; and in the second step, we
add edges $(1,2),(10,2),(12,3),(13,3),(11,9)$ and $(14,9)$. The
resulted graph $G_{\lambda_1}$ is shown in Fig. \ref{eg-fg-2}. We
can see that
$|\text{S}(G_{\lambda_1})|=|(\text{S}(G_{\lambda_0})\backslash
E)\cup\text{Out}(E)|=|\{1,4,5,6,7,8,10\}|=7<8=|\text{S}(G_{\lambda_0})|$.
So the graph in Fig. \ref{eg-fg-1} is not a minimal repair graph.
\end{exam}
\renewcommand\figurename{Fig}
\begin{figure}[htbp]
\begin{center}
\includegraphics[height=3.6cm]{fg1.12}
\end{center}
\vspace{-0.2cm}\caption{Construction of $G_{\lambda_1}$ from the
graph in Fig. \ref{eg-fg-1}.}\label{eg-fg-2}
\end{figure}

The following two corollaries give some explicit structural
properties of the minimal repair graphs of $(n,k,r,t)$-SLRC.
\begin{cor}\label{mrg-cor-1}
If $t\geq 3$, for any $v\in\text{S}(G_{\lambda_0})$, the following
hold:
\begin{itemize}
 \item [1)] $|\text{Out}(v)|\geq 1$.
 \item [2)] If $\text{Out}(v)=\{v'\}$, then $\text{Out}^2(v)=\text{Out}(v')
 \neq\emptyset$.
 \item [3)] If $\text{Out}(v)=\{v_1\}$ and $\text{Out}(v_1)=\{v_2\}$, then
 $\text{Out}(v_2)\neq\emptyset$.
 \item [4)] If $\text{Out}(v)=\{v_1\}$ and $\text{Out}(v_1)=\{v_2\}$,
 then $|\text{Out}(u)|\geq 2$ for any source
 $u\in\text{In}(v_2)$.
 \item [5)] If $v$, $w$ are two distinct sources and
 $|\text{Out}(v)|=|\text{Out}(w)|=1$, then $\text{Out}(v)\neq\text{Out}(w)$.
\end{itemize}
\end{cor}
\begin{proof}
We can prove all claims by contradiction.

1) Suppose otherwise $|\text{Out}(v)|=0$. Let $E=\{v\}$. Then,
$|\text{Out}(E)|=|\text{Out}(v)|=0<1=|\{v\}|=|E\cap\text{S}(G_{\lambda_0})|$,
which contradicts to Lemma \ref{mrg-out}.

2) Since $G_{\lambda_0}$ is acyclic and $\text{Out}(v)=\{v'\}$,
then from \eqref{Def-Out-2}, $\text{Out}^2(v)=\text{Out}(v')$. If
$\text{Out}(v')=\emptyset$, then by letting $E=\{v,v'\}$, we have
$|\text{Out}(E)|=|\emptyset|=0<1=|\{v\}|=|E\cap\text{S}(G_{\lambda_0})|$,
which contradicts to Lemma \ref{mrg-out}.

3) If $\text{Out}(v_2)=\emptyset$, then by letting
$E=\{v,v_1,v_2\}$, we have
$|\text{Out}(E)|=|\emptyset|=0<1=|\{v\}|=|E\cap\text{S}(G_{\lambda_0})|$,
which contradicts to Lemma \ref{mrg-out}.

4) By assumption, we can see that $u\neq v$. Suppose otherwise
$|\text{Out}(u)|=1$. Then $\text{Out}(u)=\{v_2\}$ since
$u\in\text{In}(v_2)$. Let $E=\{v,v_1,u\}$. We have
$|\text{Out}(E)|=|\{v_2\}|=1<2=|\{v,u\}|=|E\cap\text{S}(G_{\lambda_0})|$,
which contradicts to Lemma \ref{mrg-out}.

5) Suppose otherwise $\text{Out}(v)=\text{Out}(w)=\{v_1\}$. Let
$E=\{v,w\}$. Then we have
$|\text{Out}(E)|=|\{v_1\}|=1<2=|\{v,w\}|=|E\cap\text{S}(G_{\lambda_0})|$,
which contradicts to Lemma \ref{mrg-out}.
\end{proof}

We give in the below an example and a counterexample of minimal
repair graph that can be verified by Corollary \ref{mrg-cor-1}.
\begin{exam}
Consider the repair graph $G_{\lambda_0}$ in Fig. \ref{eg-fg-4-1},
where the vertex set is $\mathcal V=\{1,2,\cdots,15\}$. We can
check that $|\text{Out}(E)|\geq |E\cap\text{S}(G_{\lambda_0})|$
for each $E\subseteq[n]$ of size $|E|\leq t$. Corresponding to
items 1)$-$5) of Corollary \ref{mrg-cor-1}, we can check:
\begin{itemize}
 \item [1)] For every $v\in\text{S}(G_{\lambda_0})$, $|\text{Out}(v)|\geq 1$.
 \item [2)] For $v=5$ and $v'=10$, we have $\text{Out}(v)=\{v'\}$ and
 $\text{Out}^2(v)=\text{Out}(v')=\{12,13\}$.
 \item [3)] For $v=1$, $v_1=8$ and $v_2=11$, we have $\text{Out}(v)=\{v_1\}$,
 $\text{Out}(v_1)=\{v_2\}$ and $\text{Out}(v_2)=\{14\}$.
 \item [4)] For $v=1$, $v_1=8$ and $v_2=11$, we have $u=6\in\text{In}(v_2)$ is
 a source and $|\text{Out}(u)|=|\{10,11\}|\geq 2$.
 \item [5)] For $v=1$ and $w=5$, we have
 $|\text{Out}(v)|=|\text{Out}(w)|=1$ and $\text{Out}(v)=\{8\}
 \neq\text{Out}(w)=\{10\}$.
\end{itemize}
\end{exam}
\renewcommand\figurename{Fig}
\begin{figure}[htbp]
\begin{center}
\includegraphics[height=3.6cm]{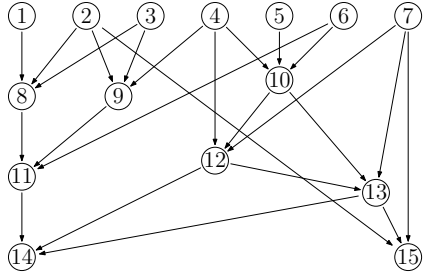}
\end{center}
\vspace{-0.2cm}\caption{An example repair graph
$G_{\lambda_0}=(\mathcal V,\mathcal E)$, where $\mathcal
V=\{1,2,\cdots,15\}$.}\label{eg-fg-4-1}
\end{figure}

\begin{exam}\label{ex-c1}
Any one of the following five observations, which violates the
corresponding five cases of Lemma \ref{mrg-cor-1}, can show that
the graph in Fig. \ref{eg-fg-1} is not a minimal repair graph.

1) For the source $v=1$, we have $\text{Out}(1)=\emptyset$.

2) For the source $v=2$, we have $\text{Out}(2)=\{9\}$ and
$\text{Out}(9)=\emptyset$.

3) For the source $v=5$, we have $\text{Out}(5)=\{11\}$,
$\text{Out}(11)=\{13\}$ and $\text{Out}(13)=\emptyset$.

4) For the source $v=6$, we have $\text{Out}(6)=\{12\}$,
$\text{Out}(12)=\{14\}$ and there is another source
$u=8\in\text{In}(14)$ such that $\text{Out}(8)=\{14\}$.

5) For the two sources $v=6$ and $w=7$, we have
$\text{Out}(6)=\text{Out}(7)=\{12\}$.
\end{exam}

\begin{rem}\label{rem-t-1-2}
In Corollary \ref{mrg-cor-1}, claim 1) holds for all $t\geq 1$,
since the contradiction is derived from a subset $E$ of size $1$.
And claims 2), 5) hold for all $t\geq 2$ since the contradictions
are derived from subsets of size $2$.
\end{rem}

\begin{cor}\label{mrg-cor-2}
Suppose $t\geq 3$ and $v\in\text{S}(G_{\lambda_0})$ such that
$\text{Out}(v)=\{v_1,v_2\}$. Then the following hold:
\begin{itemize}
 \item [1)] $\text{Out}(v_1)\neq\emptyset$ or
 $\text{Out}(v_2)\neq\emptyset$.
 \item [2)] If $\{v_1\}=\text{Out}(u)$ for some source $u$, then
 $\text{Out}(v_2)\neq\emptyset$.
 \item [3)] If $\{v_1\}=\text{Out}(u)$ for some source $u$, then
 $|\text{Out}(w)|\geq 2$ for any source $w\in\text{In}(v_2)$.
\end{itemize}
\end{cor}
\begin{proof}
All the claims can be proved by assuming the converse and choosing
a proper $E$ as in the proof of Lemma \ref{mrg-cor-1} and then
derive a contradiction.

1) Suppose otherwise $\text{Out}(v_1)=\text{Out}(v_2)=\emptyset$.
We let $E=\{v,v_1,v_2\}$ and have
$|\text{Out}(E)|=|\emptyset|=0<1=|\{v\}|=|E\cap\text{S}(G_{\lambda_0})|$,
which contradicts to Lemma \ref{mrg-out}.

2) Suppose otherwise $\text{Out}(v_2)=\emptyset$. Similarly, we can get
a contradiction by letting $E=\{u,v,v_2\}$.

3) Suppose otherwise there exist a source $w$ such that
$\text{Out}(w)=\{v_2\}$. A contradiction can be obtained by
letting $E=\{u,v,w\}$.
\end{proof}

We give in the below an example and a counterexample of minimal
repair graph that can be verified by Corollary \ref{mrg-cor-2}.
\begin{exam}
Again consider the repair graph $G_{\lambda_0}$ in Fig.
\ref{eg-fg-4-1}. Let $v=3$, $v_1=8$ and $v_2=9$. Then
$v\in\text{S}(G_{\lambda_0})$ and $\text{Out}(v)=\{v_1,v_2\}$.
Corresponding to items 1)$-$3) of Corollary \ref{mrg-cor-2}, we
can check:
\begin{itemize}
 \item [1)] $\text{Out}(v_1)=\text{Out}(v_2)=\{11\}\neq\emptyset$.
 \item [2)] For $u=1$, we have
 $\{v_1\}=\text{Out}(u)$ and $\text{Out}(v_2)=\{11\}\neq\emptyset$.
 \item [3)] For $w=4$, we can see that $w\in\text{In}(v_2)$ is a
 source and $|\text{Out}(w)|=\{9,10,12\}|\geq 2$.
\end{itemize}
\end{exam}

\begin{exam}\label{ex-c2}
Let $G$ be a repair graph as shown in Fig. \ref{eg-fg-3}. Then any
one of the following three observations, which violates the
corresponding three cases of Corollary \ref{mrg-cor-2}, can show
that $G$ is not a minimal repair graph.

1) There exists a source  $v=5$ such that $\text{Out}(v)=\{9,10\}$ and
$\text{Out}(9)=\text{Out}(10)=\emptyset$.

2) There exists a source $v=2$ such that $\text{Out}(2)=\{7,8\}$,
and a source $u=1$ such that  $\text{Out}(1)=\{7\}$ and
$\text{Out}(8)=\emptyset$.

3) There exist three sources $v=2$, $u=1$ and $w=3$ such that
$\text{Out}(2)=\{7,8\}$, $\text{Out}(1)=\{7\}$ and
$\text{Out}(3)=\{8\}$.
\end{exam}
\renewcommand\figurename{Fig}
\begin{figure}[htbp]
\begin{center}
\includegraphics[height=2.5cm]{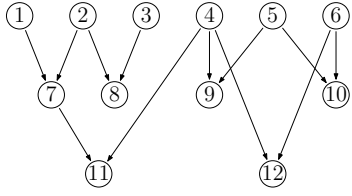}
\end{center}
\vspace{-0.2cm}\caption{An example repair graph with $n=12$
and $r=2$.}\label{eg-fg-3}
\end{figure}

\subsection{Upper Bound on the Code Rate for $(n,k,r,3)$-SLRC}
In this subsection, we assume $\mathcal C$ is an $(n,k,r,3)$-SLRC
and $G_{\lambda_0}=(\mathcal V,\mathcal E)$ is a minimal repair
graph of $\mathcal C$. Recall that $\text{S}(G_{\lambda_0})$ is
the set of all sources of  $G_{\lambda_0}$. We divide
$\text{S}(G_{\lambda_0})$ into four subsets as follows.
\begin{align}\label{divide-a}
A=\{v\in\text{S}(G_{\lambda_0}); |\text{Out}(v)|\geq
3\},\end{align}
\begin{align}\label{divide-b}B=\{v\in\text{S}(G_{\lambda_0});
|\text{Out}(v)|=2\},\end{align}
\begin{align}
\label{divide-c1}C_1=\{v\in\text{S}(G_{\lambda_0});
|\text{Out}(v)|=1 \text{~and~} |\text{Out}^2(v)|=1\}\end{align}
and
\begin{align}\label{divide-c2}C_2=\{v\in\text{S}(G_{\lambda_0});
|\text{Out}(v)|=1 \text{~and~} |\text{Out}^2(v)|\geq
2\}.\end{align} Clearly, $A, B, C_1$ and $C_2$ are mutually
disjoint. Moreover, by 1), 2) of Corollary \ref{mrg-cor-1},
$\text{S}(G_{\lambda_0})=A\cup B\cup C_1\cup C_2$. Hence,
\begin{align}\label{num-source}
\delta^*=|\text{S}(G_{\lambda_0})|=|A|+|B|+|C_1|+|C_2|.
\end{align}

We define three types of edges of $G_{\lambda_0}$, denoted by red
edge, green edge and blue edge respectively, as follows.

Firstly, an edge is called a \emph{red edge} if its tail is a
source. For each $v\in\text{S}(G_{\lambda_0})$, let $\mathcal
E_{\text{red}}(v)$ be the set of all red edges whose tail is $v$
and denote $$\mathcal
E_{\text{red}}=\bigcup_{v\in\text{S}(G_{\lambda_0})}\mathcal
E_{\text{red}}(v).$$ Then $\mathcal E_{\text{red}}$ is the set of
all red edges. Clearly, $|\mathcal
E_{\text{red}}(v)|=|\text{Out}(v)|$ and $\mathcal
E_{\text{red}}(w)\cap\mathcal E_{\text{red}}(v)=\emptyset$ for any
source $w\neq v$. So by \eqref{divide-a}$-$\eqref{divide-c2}, we
have
\begin{align}\label{num-red-edge}
|\mathcal
E_{\text{red}}|=\sum_{v\in\text{S}(G_{\lambda_0})}|\text{Out}(v)|\geq
3|A|+2|B|+|C_1|+|C_2|.
\end{align}

Secondly, an edge is called a \emph{green edge} if its tail is the
unique out-neighbor of some source in $C_1\cup C_2$. For each
$v\in C_1\cup C_2$, let $\mathcal E_{\text{green}}(v)$ be the set
of all green edges whose tail is the unique out-neighbor of $v$.
Clearly, $|\mathcal E_{\text{green}}(v)|=|\text{Out}^2(v)|$. Let
$$\mathcal E_{\text{green}}=\bigcup_{v\in C_1\cup C_2}\mathcal
E_{\text{green}}(v)$$ be the set of all green edges. Note that if
$v\neq w\in C_1\cup C_2$, then by 5) of Corollary \ref{mrg-cor-1},
$v'\neq w'$, where $v'($resp. $w')$ is the unique out-neighbor of
$v($resp. $w)$. So $\mathcal E_{\text{green}}(v)\cap\mathcal
E_{\text{green}}(w)=\emptyset$. Hence, by \eqref{divide-c1} and
\eqref{divide-c2},
\begin{align}\label{num-green-edge}
|\mathcal E_{\text{green}}|=\sum_{v\in C_1\cup
C_2}|\text{Out}^2(v)|\geq|C_1|+2|C_2|.
\end{align}

Thirdly, suppose $e\in\mathcal E$ is not a green edge and $v\in
B\cup C_1$. $e$ is called a \emph{blue edge belonging to $v$} if
one of the following two conditions hold:
\begin{itemize}
 \item[(a)] $v\in B$ and the tail of $e$ belongs to $\text{Out}(v)$.
 \item[(b)] $v\in C_1$ and the tail of $e$ belongs to $\text{Out}^2(v)$.
\end{itemize}
Let $\mathcal E_{\text{blue}}(v)$ be the set of all blue edges
belonging to $v$ and let $$\mathcal E_{\text{blue}}=\bigcup_{v\in
B\cup C_1}\mathcal E_{\text{blue}}(v)$$ be the set of all blue
edges. Then we have the following lemma.
\begin{lem}
The number of blue edges is lower bounded by
\begin{align}\label{num-blue-edge}
|\mathcal E_{\text{blue}}|\geq\frac{|B|+|C_1|}{r}.
\end{align}
\end{lem}

\begin{proof}
It is sufficient to prove : i) For each $v\in B\cup C_1$, there
exists at least one blue edge belonging to $v$; and ii) Each blue
edge belongs to at most $r$ different $v\in B\cup C_1$. To prove
these two statements, we will use the definition of red edge,
green edge and blue edge repeatedly.

We first prove i) by considering the cases of $v\in B$ and $v\in
C_1$.

Let $v\in B$, and we look for a blue edge belonging to $v$. In
this case, by \eqref{divide-b}, we can assume
$\text{Out}(v)=\{v_1,v_2\}~($see Fig. \ref{fg-clm-1}(a)$)$. Then,
by 1) of Corollary \ref{mrg-cor-2}, $\text{Out}(v_1)\neq\emptyset$
or $\text{Out}(v_2)\neq\emptyset$. Without loss of generality,
assume $\text{Out}(v_1)\neq\emptyset$ and $v_3\in\text{Out}(v_1)$.
Consider $(v_1,v_3)$. If it is not a green edge, then by
definition, it is a blue edge belonging to $v$. So we assume that
$(v_1,v_3)$ is a green edge. Then by definition,
$\{v_1\}=\text{Out}(u)$ for some $u\in C_1\cup C_2$. By 2) of
Corollary \ref{mrg-cor-2}, $\text{Out}(v_2)\neq\emptyset$ and we
can let $v_4\in\text{Out}(v_2)$, as illustrated in Fig.
\ref{fg-clm-1}(a). Consider $(v_2,v_4)$. By 3) of Corollary
\ref{mrg-cor-2}, $|\text{Out}(w)|\geq 2$ for any source
$w\in\text{In}(v_2)$, which implies $(v_2,v_4)$ is not a green
edge. So $(v_2,v_4)$ is a blue edge belonging to $v$. Hence, for
each $v\in B$, we can always find a blue edge belonging to $v$.

\renewcommand\figurename{Fig}
\begin{figure}[htbp]
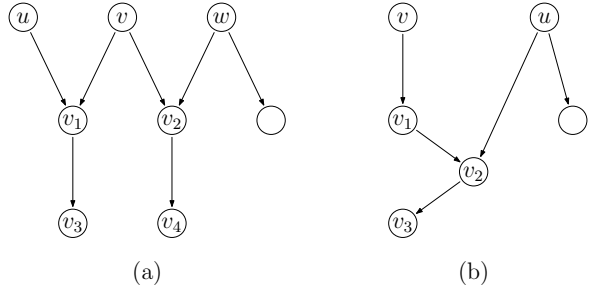

\begin{center}
\includegraphics[height=3.8cm]{fg1.2}
\hspace{1.2cm}\includegraphics[height=3.8cm]{fg1.1}
\end{center}
\vspace{-0.2cm}\caption{Two local graphs.}\label{fg-clm-1}
\end{figure}

Now, let $v\in C_1$ and we look for a blue edge belonging to $v$.
By \eqref{divide-c1}, we can assume $\text{Out}(v)=\{v_1\}$ and
$\text{Out}^2(v)=\{v_2\}~($see Fig.\ref{fg-clm-1}(b)$)$. By 3) of
Corollary \ref{mrg-cor-1}, we have $\text{Out}(v_2)\neq\emptyset$.
Let $v_3\in\text{Out}(v_2)$. Note that by 4) of Corollary
\ref{mrg-cor-1}, $|\text{Out}(u)|\geq 2$ for any source
$u\in\text{In}(v_2)$ (see Fig. \ref{fg-clm-1}(b) as illustration),
which implies $(v_2,v_3)$ is not a green edge. So $(v_2,v_3)$ is a
blue edge belonging to $v$. Hence, for each $v\in C_1$, we can
always find a blue edge belonging to $v$.

By the above discussion, statement i) holds.

Let $(u',u'')$ be a blue edge and $S$ be the set of all $v\in
B\cup C_1$ such that $(u',u'')$ belongs to $v$. To prove statement
ii), we prove that there is an {\em injection}, namely $\varphi$,
from $S$ to $\text{In}(u')$. Then ii) follows from the fact that
$\text{In}(u')$ has size at most $r$. The injection of
$\varphi(v)$ can be constructed as follows: If $v\in B$, simply
let $\varphi(v)=v$. If $v\in C_1$, let $\varphi(v)=v'$, where
$\{v'\}=\text{Out}(v)$. It is easy to see that $\varphi(v)$ is an
injection (noticing 5) of Corollary \ref{mrg-cor-1}), which
completes the proof of statement ii).
\end{proof}

\begin{exam}
Consider the repair graph in Fig. \ref{eg-fg-4-1}. We have
$A=\{2,4,7\}$, $B=\{3,6\}$, $C_1=\{1\}$ and $C_2=\{5\}$, and the
edges with tails from 1 to 7 are red edges, as illustrated in Fig.
\ref{eg-fg-4}.

Moreover, one can check that $\mathcal
E_{\text{green}}(1)=\{(8,11)\}$ and $\mathcal
E_{\text{green}}(5)=\{(10,12), (10,13)\}$. As for blue edges,
since $1\in C_1$ and $11\in\text{Out}^2(1)$, then
$(11,14)\in\mathcal E_{\text{blue}}(1)$; Since
$11\in\text{Out}(6)$ and $6\in B$, then $(11,14)\in\mathcal
E_{\text{blue}}(6)$; Since $3\in B$ and $9\in\text{Out}(3)$, then
$(9,11)\in\mathcal E_{\text{blue}}(3)$. One can check that
$\mathcal E_{\text{blue}}(1)=\mathcal
E_{\text{blue}}(6)=\{(11,14)\}$ and $\mathcal
E_{\text{blue}}(3)=\{(9,11)\}$. The green edges and blue edges are
also illustrated in Fig. \ref{eg-fg-4}.
\end{exam}

\renewcommand\figurename{Fig}
\begin{figure}[htbp]
\begin{center}
\includegraphics[height=3.6cm]{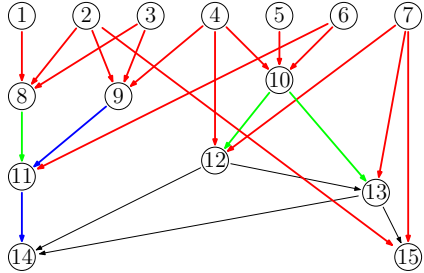}
\end{center}
\vspace{-0.2cm}\caption{Illustration of red edge, green edge and
blue edge of minimal repair graph.}\label{eg-fg-4}
\end{figure}

Now, we can propose our main theorem of this section.
\begin{thm}\label{bnd-t-3}
For $(n,k,r,3)$-SLRC, we have \footnote{In the original version
\cite{Wentu-15} of this paper, bound \eqref{rate-t-is-3} was
presented equivalently in terms of the code length as $n\geq
k+\left\lceil\frac{2k+\lceil\frac{k}{r}\rceil}{r}\right\rceil$.}
\begin{align}\label{rate-t-is-3} \frac{k}{n}\leq
\left(\frac{r}{r+1}\right)^2.\end{align}
\end{thm}

\begin{proof}
By definition, we can easily see that $\mathcal
E_{\text{red}},\mathcal E_{\text{green}}$ and $\mathcal
E_{\text{blue}}$ are mutually disjoint. Then by
\eqref{num-source}-\eqref{num-blue-edge}, we have
\begin{align*}
|\mathcal E|&\geq|\mathcal E_{\text{red}}|+|\mathcal
E_{\text{green}}|+|\mathcal E_{\text{blue}}|\\&\geq
(3|A|+2|B|+|C_1|+|C_2|)\\& ~ ~ ~ +(|C_1|+2|C_2|)+\frac{|B|+|C_1|}{r}\\
&=
2(|A|+|B|+|C_1|+|C_2|)\\& ~ ~ ~ +(|A|+|C_2|+\frac{|B|+|C_1|}{r})\\
&=
2\delta^*+\frac{r|A|+r|C_2|+|B|+|C_1|}{r}\\
&\geq
2\delta^*+\frac{|A|+|C_2|+|B|+|C_1|}{r}\\
&=2\delta^*+\frac{\delta^*}{r}.
\end{align*}
That is, $|\mathcal E|\geq 2\delta^*+\frac{\delta^*}{r}$.
Combining this with Lemma \ref{edge-num-left}, we have
\begin{align*}(n-\delta^*)r\geq |\mathcal E| \geq
2\delta^*+\frac{\delta^*}{r}.\end{align*}
So
\begin{align*}(n-\delta^*)r\geq
2\delta^*+\frac{\delta^*}{r}.\end{align*}
Solving $n$ from the above equation, we have
\begin{align}\label{add-t-is-3} n\geq\delta^*+\frac{2\delta^*+
\frac{\delta^*}{r}}{r} .\end{align} By Lemma \ref{dim-dlt},
$\delta^*\geq k$. So \eqref{add-t-is-3} implies that
\begin{align*}n&\geq k+\frac{2k+
\frac{k}{r}}{r}\\&=k\left(1+\frac{2}{r}+\frac{1}{r^2}\right)
\\&=k\left(\frac{r+1}{r}\right)^2.\end{align*}
Hence,
\begin{align*} \frac{k}{n}\leq
\left(\frac{r}{r+1}\right)^2,\end{align*} which proves the theorem.
\end{proof}

We will later construct two families of $(n,k,r,3)$-SLRCs
achieving \eqref{rate-t-is-3} and hence show the tightness of this
bound.

\subsection{Code Rate for $(n,k,r,2)$-SLRC}
In this subsection, we give a new proof of the bound
\eqref{rate-bd-3} for the $(n,k,r,2)$-SLRC using the similar
techniques as in Subsection B. Assume that $\mathcal C$ is an
$(n,k,r,2)$-SLRC and $G_{\lambda_0}=(\mathcal V,\mathcal E)$ is a
minimal repair graph of $\mathcal C$.

\begin{proof}[Proof of Bound \eqref{rate-bd-3}]
By Remark \ref{rem-t-1-2} and 1) of Corollary \ref{mrg-cor-1},
each source of $G_{\lambda_0}$ has at least one out-neighbor. Let
$A$ be the set of sources that has only one out-neighbor and let
$\mathcal E_{\text{red}}$ be the set of all edges $e$, called red
edges, such that the tail of $e$ is a source. Then the number of
red edges is
\begin{align}\label{2t-red}
|\mathcal
E_{\text{red}}|\geq|A|+2|\text{S}(G_{\lambda_0})\backslash
A|=2\delta^*-|A|.\end{align}

For each $v\in A$, let $v'$ be the unique out-neighbor of $v$ and
$\mathcal E_{\text{green}}(v)$ be the set of all edges whose tail
is $v'$. By Remark \ref{rem-t-1-2} and 2) of Corollary
\ref{mrg-cor-1}, $\text{Out}^2(v)=\text{Out}(v')\neq\emptyset$. So
$|\mathcal E_{\text{green}}(v)|=|\text{Out}(v')|\geq 1$. Let
$\mathcal E_{\text{green}}$ be the set of all green edges. For any
two different $v_1,v_2\in A$, let $v_1', v_2'$ be the unique
out-neighbor of $v_1,v_2$, respectively. By Remark \ref{rem-t-1-2}
and 5) of Corollary \ref{mrg-cor-1}, $v_1'\neq v_2'$. So $\mathcal
E_{\text{green}}(v_1)\cap\mathcal
E_{\text{green}}(v_2)=\emptyset$. Hence,
\begin{align}\label{2t-green}|\mathcal
E_{\text{green}}|=\left|\bigcup_{v\in A}\mathcal
E_{\text{green}}(v)\right|=\sum_{v\in A}\left|\mathcal
E_{\text{green}}(v)\right|\geq|A|.\end{align}

Clearly, $\mathcal E_{\text{red}}\cap \mathcal
E_{\text{green}}=\emptyset$. Then by \eqref{2t-red} and
\eqref{2t-green}, $$|\mathcal E|\geq|\mathcal
E_{\text{red}}|+|\mathcal E_{\text{green}}|\geq 2\delta^*.$$ On
the other hand, by Lemma \ref{edge-num-left},
$$(n-\delta^*)r\geq|\mathcal E|.$$ So
$(n-\delta^*)r\geq 2\delta^*$, which implies $n\geq
\delta^*+\frac{2\delta^*}{r}\geq k+\frac{2k}{r}~($Lemma
\ref{dim-dlt}$)$. Hence, $\frac{k}{n}\leq\frac{r}{r+2}$, which
completes the proof.
\end{proof}

\renewcommand\figurename{Fig}
\begin{figure*}[htbp]
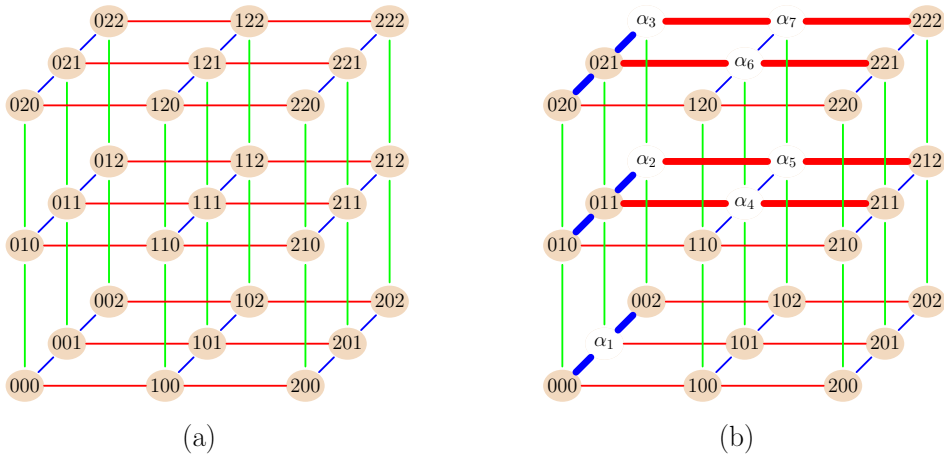

\begin{center}
\includegraphics[height=6cm]{itrust-fig.5}
\hspace{1.52cm}
\includegraphics[height=6cm]{itrust-fig.7}
\end{center}
\caption{(a) The index set of $\mathbb Z_3^3$, where
$(i_3,i_2,i_1)$ is simply written as $i_3i_2i_1$; (b) The
recovering sets $R_1,\!\cdots\!,R_7$ of
$\alpha_1,\!\cdots\!,\alpha_7$, where $R_1\!=\!\{(000),(002)\}$,
$R_2\!=\!\{(010),(011)\}$, $R_3\!=\!\{(020),(021)\}$,
$R_4\!=\!\{(011),(211)\}$, $R_5\!=\!\{\alpha_2,(212)\}$,
$R_6\!=\!\{(021),(221)\}$ and $R_7\!=\!\{\alpha_3,(222)\}$ drawn
in heavy lines.} \label{exam-code-1}
\end{figure*}

\section{An Example of SLRC}
In order to have a better understanding of the binary $(r,t)$-SLRC
constructed in the next section, we give an example in this
section. Let $r=2$, $m=3$ and $\mathcal C$ be the product code of
$m$ copies of the binary $[r+1, r]$ single parity check code. Then
$\mathcal C$ has length $n=(r+1)^3=27$ and dimension $k=r^3=8$. It
is convenient to use $\mathbb Z_{3}^3=\{(i_3,i_2,i_1);
i_1,i_2,i_3\in\mathbb Z_3\}$ instead of $[n]$ as the index set of
the coordinates of $\mathcal C$, and let $\mathbb
Z_{2}^3=\{(i_3,i_2,i_1); i_1,i_2,i_3\in\mathbb Z_2\}$ be the
information set of $\mathcal C$. Here, $\mathbb Z_3=\{0,1,2\}$ and
$\mathbb Z_2=\{0,1\}$ are simply viewed as two sets and $\mathbb
Z_2\subseteq \mathbb Z_3~($no algebraic meaning is considered
here$)$.

The index set $\mathbb Z_3^3$ is depicted in Fig.
\ref{exam-code-1}(a). By definition, each code symbol (coordinate)
of $\mathcal C$ can be recovered by all the other symbols on the
same (red, green or blue) line. Hence, each code symbol of
$\mathcal C$ has $m=3$ disjoint recovering sets (red, green and
blue) of size $r=2$, and $\mathcal C$ can recover any $3$ erasures
by parallel recovery. However, we can prove ({\em see details in
the next section}) that it can recover any $t=2^m-1=7$ erasures by
sequential recovery. For example, consider an erasure of $7$ code
symbols, say, $E=\{\alpha_1,\cdots,\alpha_7\}$, where
$\alpha_1=(001)$, $\alpha_2=(012)$, $\alpha_3=(022)$,
$\alpha_4=(111)$, $\alpha_5=(112)$, $\alpha_6=(121)$ and
$\alpha_7=(122)$, as illustrated in Fig. \ref{exam-code-1}(b). We
can select a sequence of recovering sets $R_1=\{(000),(002)\}$,
$R_2=\{(010),(011)\}$, $R_3=\{(020),(021)\}$,
$R_4=\{(011),(211)\}$, $R_5=\{\alpha_2,(212)\}$,
$R_6=\{(021),(221)\}$ and $R_7=\{\alpha_3,(222)\}~($see Fig.
\ref{exam-code-1}(b)$)$. It is easy to check that $R_1,\cdots,R_7$
sequentially repair $\{\alpha_1,\cdots,\alpha_7\}$.

In general, by puncturing $\mathcal C$ properly, we can obtain
$(r,t)$-SLRC for any $t\in\{1,2,\cdots,6\}$. As an example, we
construct an $(r,5)$-SLRC as follows. For each $j\in \mathbb
Z_2=\{0,1\}$, let
$$A_j=\{(j,i_2,i_1); i_2,i_1\in\mathbb Z_3\}$$ and let $$A=A_0\cup
A_1,$$
\begin{align*}B&=\{(2,i_2,i_1); i_2\in\mathbb Z_2 ~\text{and}~
i_1\in\mathbb
Z_3\}\\&=\{(200),(201),(202),(210),(211),(212)\}.\end{align*} Let
$\Omega=A\cup B$, as depicted in Fig. \ref{exam-code-2}. Then, the
punctured code $\mathcal C|_{\Omega}$ is an $(n',k,r,5)$-SLRC,
with $n'=|\Omega|=24$.

In fact, one can see that the following items hold.
\begin{description}
  \item[i)] For any nonempty $E\subseteq A$ of size
$|E|\leq t_1=3$, there exists an $\alpha\in E$ such that $\alpha$
has a recovering set $R\subseteq A\backslash E$.
  \item[ii)] For any nonempty $E\subseteq A$ of size
$|E|\leq t=5$, there exists an $\alpha\in E$ such that $\alpha$
has a recovering set $R\subseteq\Omega\backslash E$.
  \item[iii)] For any nonempty $E\subseteq B$ of size
$|E|\leq t_2=1$, there exists an $\alpha\in E$ such that $\alpha$
has a recovering set $R\subseteq B\backslash E$.
  \item[iv)] For any nonempty $E\subseteq B$ of size
$|E|\leq t=5$, there exists an $\alpha\in E$ such that $\alpha$
has a recovering set $R\subseteq\Omega\backslash E$.
\end{description}

In the above, items $i),iii),iv)$ can be easily verified. For
example, one can see that the punctured codes $\mathcal C|_{A_0}$
and $\mathcal C|_{A_1}$ are both $(r,3)$-SLRC and $\mathcal
C|_{B}$ is a $(r,1)$-SLRC, hence i) and iii) hold. From Fig.
\ref{exam-code-2}, one can see that each $(2,i_2,i_1)\in B$ has a
recovery set (red line) $R=\{(0,i_2,i_1), (1,i_2,i_1)\}\subseteq
A$, hence iv) holds. To prove ii), we consider the following two
cases:
\begin{enumerate}
  \item [1)] $E\subseteq A_0$ or $E\subseteq A_1$. Without loss of
generality, assume $E\subseteq A_1$. If
$E\subseteq\{(120),(121),(122)\}$, then each $(1,2,i)\in E$ has a
recovery set (green line) $R=\{(1,0,i),
(1,1,i)\}\subseteq\Omega\backslash E$; Otherwise, there exists a
$(1,i_2,i_1)\in E\cap\{(100),(101),(102),(110),(111),(112)\}$
which has a recovery set (red line) $R=\{(0,i_2,i_1),
(2,i_2,i_1)\}\subseteq\Omega\backslash E$.
   \item [2)] $E\cap
A_0\neq \emptyset$ and $E\cap A_1\neq \emptyset$. Since $|E|\leq
5$, then $|E\cap A_0|\leq 3$ or $|E\cap A_1|\leq 3$. Note that
both $\mathcal C|_{A_0}$ and $\mathcal C|_{A_1}$ are
$(r,3)$-SLRCs, by Lemma \ref{lem-ELRC}, there exists an $\alpha\in
E$ and $j\in\{0,1\}$ such that $\alpha$ has a recovering set
$R\subseteq A_j\backslash E\subseteq A\backslash
E\subseteq\Omega\backslash E$.
\end{enumerate}
Then, by Lemma \ref{lem-Com-LRC}, $\mathcal C|_{\Omega}$ is an
$(n',k,r,5)$-SLRC. The generalization of this example as well as
the formal proof will be given in the next section.

\renewcommand\figurename{Fig}
\begin{figure}[htbp]
\begin{center}
\includegraphics[height=5.45cm]{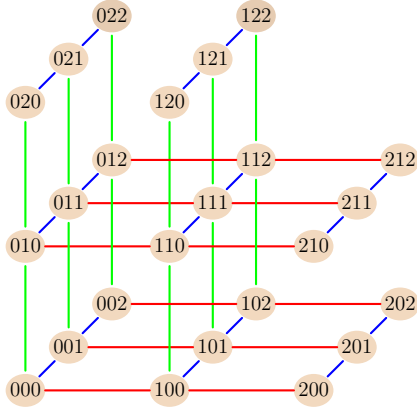}
\end{center}
\caption{Graphical illustration of a subset of $\mathbb Z_3^3$.}
\label{exam-code-2}
\end{figure}

\section{Construction of $(n,k,r,t)$-SLRC}
In this section, we construct a family of binary $(n,k,r,t)$-SLRC
for any positive integers $r~(\geq 2)$ and $t$. It will be shown
that the code rate of this family is greater than $\frac{r}{r+t}$,
and in particular, for $t\in\{2,3\}$, it achieves the bounds
\eqref{rate-bd-3} and \eqref{rate-t-is-3}, respectively.

We first need introduce some notations. For any positive integers
$r$ and $m$, where $r\geq 2$, let $\mathbb
Z_r=\{0,1,\cdots\!,r-1\}$ and $\mathbb
Z_r^m=\{(i_m,i_{m-1},\cdots\!,i_1);
i_m,i_{m-1},\cdots\!,i_1\in\mathbb Z_r\}$. Here, $\mathbb Z_r$ is
simply viewed as a set $($without any algebraic meaning$)$. So
$\mathbb Z_r\subseteq\mathbb Z_{r+1}=\{0,1,\cdots, r-1, r\}$. We
will use $\alpha$, $\beta$, $\gamma$, etc, to denote elements
(points) of $\mathbb Z_{r+1}^m$. Note that by the notation, for
each $\alpha=(i_m,i_{m-1},\cdots\!,i_1)\in\mathbb Z_r$ and
$\ell\in[m]$, $i_\ell$ is the $\ell$th coordinate of $\alpha$ from
the right.

For each $\alpha=(i_m,i_{m-1},\cdots\!,i_1)\in\mathbb Z_{r+1}^m$,
we let
\begin{align}\label{U-alf}
\text{\textbf{U}}^{(m)}(\alpha)=\{\ell\in[m]; i_\ell=r\},
\end{align} and
\begin{align}\label{T-alf}
\text{\textbf{T}}^{(m)}(\alpha)=\{\ell\in[m]; i_\ell\in\mathbb
Z_r\}.
\end{align}
Further, we let
\begin{align}\label{L-alf}
\mathcal
L^{(m)}(\alpha)\!=\!\{(j_m,j_{m-1},\cdots\!,j_1)\!\in\!\mathbb
Z_{r}^m; j_\ell\!=\!i_\ell, \forall \ell\!\in\!
\text{\textbf{T}}^{(m)}(\alpha)\}.\end{align}

Clearly,
$\text{\textbf{U}}^{(m)}(\alpha)\cap\text{\textbf{T}}^{(m)}
(\alpha)=\emptyset$ and $\text{\textbf{U}}^{(m)}(\alpha)\cup
\text{\textbf{T}}^{(m)}(\alpha)=[m]$. Moreover, for each
 $\alpha\in\mathbb Z_{r+1}^m\backslash\mathbb Z_{r}^m$,
 $\text{\textbf{U}}^{(m)}(\alpha)\neq\emptyset$ and $\mathcal
 L^{(m)}(\alpha)\neq\emptyset$. In particular, if
 $\alpha=(r,r,\cdots,r)$, then $\text{\textbf{U}}^{(m)}(\alpha)=[m]$
 and $\mathcal L^{(m)}(\alpha)=\mathbb Z_r^m$.

As an example, let $r=2$, $m=6$ and
$\alpha=(1,0,2,1,1,2)\in\mathbb Z_3^6$. Then
$\text{\textbf{U}}^{(m)}(\alpha)=\{4,1\}$,
$\text{\textbf{T}}^{(m)}(\alpha)=\{6,5,3,2\}$ and $\mathcal
L^{(m)}(\alpha)=\{(1,0,i_4,1,1,i_1); i_4,i_1\in\mathbb
Z_2\}=\{(1,0,0,1,1,0), (1,0,0,1,1,1), (1,0,1,1,1,0),
(1,0,1,1,1,1)\}$.

Note that for any integer $s$ such that $0\leq s\leq 2^{m}-1$, $s$
has a unique $m$-digit binary representation, say
$(\lambda_m\lambda_{m-1}\cdots\lambda_1)$. That is,
$(\lambda_m,\lambda_{m-1},\cdots\!,\lambda_1)\!\in\!\{0,1\}^m$ and
$s=\sum_{\ell=1}^m\lambda_{\ell}2^{\ell-1}$. Denote by
$\text{supp}_m(s)$ the support of
$(\lambda_m,\lambda_{m-1},\cdots\!,\lambda_1)$. Let
\begin{align}\label{def-GM}
\Gamma_s^{(m)}\!=\!\{\alpha\!\in\!\mathbb Z_{r+1}^m;
\text{\textbf{U}}^{(m)}(\alpha)\!=\!\text{supp}_m(s)\}\end{align}
and
\begin{align}\label{def-OM}
\Omega_s^{(m)}=\bigcup_{\ell=0}^s\Gamma_{\ell}^{(m)}.\end{align}

For example, suppose $r=2$, $m=6$ and $s=22$. Then $(010110)$ is
the unique $6$-digit binary representation of $s$ and
$\text{supp}_m(s)=\{5,3,2\}$. From \eqref{def-GM}, we have
$\Gamma_{22}^{(6)}=\{(i_6,2,i_4,2,2,i_1); i_6,i_4,i_1\in\mathbb
Z_2\}$.

Clearly, $\Gamma_0^{(m)}, \Gamma_1^{(m)}, \cdots,
\Gamma_{2^m-1}^{(m)}$ are mutually disjoint and
$|\Gamma_s^{(m)}|=r^{m-|\text{supp}_m(s)|}$, for $s=0,1,
 \cdots, 2^m-1$. In particular, $|\Gamma_0^{(m)}|=r^m$ and
$|\Gamma_{2^m-1}^{(m)}|=1$. Moreover, by definition, we have
$\Omega_0^{(m)}=\Gamma_0^{(m)}=\mathbb
 Z_r^m$ and $\Omega_{2^m-1}^{(m)}=\mathbb Z_{r+1}^m$.

For any positive integers $r~(\geq2)$ and $t$, we can always pick
an integer $m$ such that $t\leq 2^m-1$ and, using the above
notations, define a matrix $H_t^{(m)}=(h_{\alpha,\beta})$
satisfying the following two properties.
\begin{itemize}
 \item [(1)] The rows of $H_t^{(m)}$ are indexed by
 $\Omega_t^{(m)}\backslash\Omega_0^{(m)}$ and the columns of
 $H_t^{(m)}$ are indexed by $\Omega_t^{(m)}$; \
 \item [(2)] For each $\alpha\in\Omega_t^{(m)}\backslash\Omega_0^{(m)}$
 and $\beta\in\Omega_t^{(m)}$,
 \begin{equation}\label{def-H}
 h_{\alpha,\beta}=\left\{\begin{aligned}
 &1, ~ ~\text{if}~\beta\in\mathcal L^{(m)}(\alpha)\cup\{\alpha\};\\
 &0, ~ ~\text{Otherwise}.\\
 \end{aligned} \right.
 \end{equation}
 \end{itemize}

It should be noted that, $H_t^{(m)}=(h_{\alpha,\beta})$ is an
$h\times n$ {\em binary} matrix, where
$h=\left|\Omega_t^{(m)}\backslash\Omega_0^{(m)}\right|$,
$n=\left|\Omega_t^{(m)}\right|$. The sub-matrix of $H_t^{(m)}$,
formed by the columns indexed by
$\Omega_t^{(m)}\backslash\Omega_0^{(m)}$, is a permutation matrix.
Hence, $\text{rank}\left(H_t^{(m)}\right)=h$.

\begin{thm}\label{main-th}
Let $\mathcal C_t^{(m)}$ be the binary code that has a parity
check matrix $H_t^{(m)}$. Then $\mathcal C_t^{(m)}$ is an
$(n,k,r,t)$-SLRC with
\begin{align}\label{n-value}
n=r^m\sum_{s=0}^t\frac{1}{r^{|\text{supp}_m(s)|}}
\end{align}
and \begin{align}\label{k-value} k=r^m.
\end{align} Hence, the code rate of $\mathcal
C_t^{(m)}$ is
\begin{align}\label{n-k-value}
\frac{k}{n}=\frac{1}{\sum_{s=0}^t\frac{1}{r^{|\text{supp}_m(s)|}}},
\end{align}
where $r~(\geq 2)$ and $t$ are any positive integers and $m$ is
any integer satisfying $t\leq 2^m-1$.
\end{thm}

\begin{rem}
We have some remarks about the construction.
\begin{itemize}
 \item [1)] The example codes given in the last
section are just $\mathcal C_t^{(m)}$ for $r=2$, $m=3$ and $t=7,5$
respectively. In general, for $t=2^m-1$, it is easy to check that
$\Omega_{2^m-1}^{(m)}=\mathbb Z_{r+1}^m$ and $\mathcal
C_{2^m-1}^{(m)}$ is the product of $m$ copies of the $[r+1,r]$
binary code. If $t<2^m-1$, then $\mathcal C_t^{(m)}$ is the
punctured code of $\mathcal C_{2^m-1}^{(m)}$ with respect to
$\Omega_t^{(m)}$.
 \item [2)] For $t\in\{2,3\}$, we can let $m=2$ and from
\eqref{n-k-value}, the code rates of our construction are
$\frac{r}{r+2}$ and $\left(\frac{r}{r+1}\right)^2$ respectively,
which are optimal according to \eqref{rate-bd-3} and
\eqref{rate-t-is-3}. For $t\geq 4$, by \eqref{n-k-value}, the code
rate of $\mathcal C_t^{(m)}$ is higher than $\frac{r}{r+t}$ for
all $r\geq 2$.
 \item [3)] It was shown in \cite{Tamo14} that
$\mathcal C_{2^m-1}^{(m)}$ has locality $r$ and availability $m$,
which implies that it can recover $m$ erasures with locality $r$
using the parallel approach. In contrast, by Theorem
\ref{main-th}, it can recover $t=2^m-1$ erasures with the same
locality when using the sequential approach, which is a
significant advantage of the product code for the sequential
recovery. In particular, the product of two copies of the
$[r+1,r]$ binary code is not optimal (in rate) among codes with
locality $r$ and availability $t=2$ \cite{Wang15}, but optimal
among $(r,t=3)$-SLRCs.
\end{itemize}
\end{rem}

In the rest of this section, we will prove Theorem \ref{main-th}.
To prove that $\mathcal C_t^{(m)}$ is an $(r,t)$-SLRC, we will
prove a more general claim, say, for any binary linear code
$\mathcal C$, if $\mathcal C$ has a parity check matrix $H$ which
contains all rows of $H_t^{(m)}~($not necessarily $H=H_t^{(m)})$,
then $\mathcal C$ is an $(r,t)$-SLRC. We first make some
clarifications on the construction by two simple remarks.

\begin{rem}\label{rem-pcm}
Let $\mathcal C$ be a binary linear code. If the code symbols of
$\mathcal C$ are indexed by $\Omega_t^{(m)}$, then, by
construction of $H_t^{(m)}$, $\mathcal C$ has a parity check
matrix which contains all rows of $H_t^{(m)}$ if and only if for
each $\alpha\in\Omega_{t}^{(m)}\backslash\Omega_{0}^{(m)}$,
\begin{align}\label{symbl-rltn}
x_\alpha=\sum_{\beta\in\mathcal
L^{(m)}(\alpha)}x_{\beta}.\end{align} If the code symbols of
$\mathcal C$ are indexed by $S$, where $S\neq\Omega_t^{(m)}$, then
$\mathcal C$ has a parity check matrix which contains all rows of
$H_t^{(m)}$ if and only if there is a bijection $\psi:
\Omega_t^{(m)}\rightarrow S$ such that for each
$\alpha\in\Omega_{t}^{(m)}\backslash\Omega_{0}^{(m)}$,
\begin{align}\label{symbl-rltn-not-OM}
x_{\psi(\alpha)}=\sum_{\beta\in\mathcal
L^{(m)}(\alpha)}x_{\psi(\beta)}.\end{align}
\end{rem}

\begin{rem}\label{prtn-n}
Since $1\leq t\leq 2^m-1$, we can find a $m_0\in[m]$ such that
$2^{m_0-1}-1<t\leq2^{m_0}-1$. Let $t_1=2^{m_0-1}-1$ and
$t_2=t-t_1-1$. Then $0\leq t_2\leq t_1\leq 2^{m-1}-1$ and
$\Omega_{t}^{(m)}$ can be partitioned into two disjoint nonempty
subsets
$$A=\Omega_{t_1}^{(m)}=\bigcup_{s=0}^{t_1}\Gamma_{s}^{(m)}$$ and
$$B=\Omega_{t}^{(m)}\backslash A
=\bigcup_{s=t_1+1}^t\Gamma_{s}^{(m)}.$$ Moreover, noticing that
$\text{supp}_m(s)\subseteq\{1,2,\cdots\!,m_0-1\}$ for $0\leq s\leq
t_1$, then $A$ can be partitioned into $r$ mutually disjoint
nonempty subsets, according to the values of the $m_0$th
coordinate $($from the right$)$ of its elements, as follows.
$$A_j=\{(i_m,i_{m-1},\cdots,i_1)\in A;~ i_{m_0}\!=j\},~ \forall
j\in\mathbb Z_r.$$
\end{rem}

In the following, if there is no other specification, we always
assume that $\mathcal C$ is a binary linear code and has a parity
check matrix which contains all rows of $H_t^{(m)}$. Without loss
of generality, we assume that the code symbols of $\mathcal C$ are
indexed by $\Omega_t^{(m)}$. To prove Theorem \ref{main-th}, we
need the following three lemmas.

\begin{lem}\label{Parity-proj}
Suppose $m>1$. With notations in Remark \ref{prtn-n}, the
following hold.
\begin{itemize}
 \item [1)] For each $j\!\in\!\mathbb Z_r$, the punctured
 code $\mathcal C|_{A_j}$ has a parity check matrix which contains
 all rows of $H_{t_1}^{(m-1)}$.
 \item [2)] If $t_2\geq 1$, the punctured code $\mathcal C|_B$
 has a parity check matrix which contains all rows of
 $H_{t_2}^{(m-1)}$.
\end{itemize}
\end{lem}
\begin{proof}
For each $j\in\mathbb Z_{r+1}=\{0,1,\cdots,r\}$, let
$$\psi_j: \mathbb Z_{r+1}^{m-1}\rightarrow \mathbb Z_{r+1}^{m}$$
be such that $\psi_j(\alpha)\!=\!(i_{m-1},\!\cdots\!,i_{m_0}, j,
i_{m_0-1},\!\cdots\!, i_{1})$ for each $\alpha\!\in\!(i_{m-1},
\!\cdots\!,i_{m_0}, i_{m_0-1},\!\cdots\!, i_{1})\!\in\!\mathbb
Z_{r+1}^{m-1}$. That is, $\psi_j(\alpha)$ is obtained by inserting
$j$ as a coordinate between the $(m_0\!-\!1)$th and $m_0$th
coordinate $($from the right$)$ of $\alpha$.

1) For each $j\in\mathbb Z_r$,
it is a mechanical work to check that $\psi_j$ induces a bijection
between $\Omega_{t_1}^{(m-1)}$ and $A_j~$ such that for each
$\alpha\in\Omega_{t_1}^{(m-1)}\backslash\Omega_{0}^{(m-1)}$,
$$\mathcal L^{(m)}(\psi_j(\alpha))=\{\psi_j(\beta); \beta\in\mathcal
L^{(m-1)}(\alpha)\}.$$ Since $\mathcal C$ has a parity check
matrix containing all rows of $H_t^{(m)}$, then by
\eqref{symbl-rltn}, we have
\begin{align*}x_{\psi_j(\alpha)}&=\sum_{\beta'\in\mathcal
L^{(m)}(\psi_j(\alpha))}x_{\beta'}\\&=\sum_{\beta\in\mathcal
L^{(m-1)}(\alpha)}x_{\psi_j(\beta)}.\end{align*} Hence, by Remark
\ref{rem-pcm}, $\mathcal C|_{A_j}$ has a parity check matrix which
contains all rows of $H_{t_1}^{(m-1)}$.

2) Recall that $t_1=2^{m_0-1}-1$. Then for each
$s\in\{t_1+1,t_1+2,\cdots, t\}$, we have
$$\text{supp}_m(s)=\text{supp}_{m-1}(s')\cup\{m_0\},$$ where
$s'=s-2^{m_0-1}\in\{0,1,\cdots,t_2\}$. So similar to 1), we can
check that $\psi_r$ induces a bijection between
$\Omega_{t_2}^{(m-1)}$ and
$B=\Omega_{t}^{(m)}\backslash\Omega_{t_1}^{(m)}$ such that for
each $\alpha\in\Omega_{t_2}^{(m-1)}\backslash\Omega_{0}^{(m-1)}$,
$$x_{\psi_r(\alpha)}=\sum_{\beta\in\mathcal
L^{(m-1)}(\alpha)}x_{\psi_r(\beta)}.$$ Hence, by Remark
\ref{rem-pcm}, $\mathcal C|_B$ has a parity check matrix which
contains all rows of $H_{t_2}^{(m-1)}$.
\end{proof}

For each $\alpha\!=\!(i_m,i_{m-1},\cdots\!,i_1)\!\in\!\mathbb
Z_{r+1}^m$ and $\ell\!\in\![m]$, let \begin{align}\label{mL-alf}
L^{(\ell)}_\alpha\!=\!\{(i_m,\cdots\!,i_{\ell+1}, i_\ell',
i_{\ell-1},\cdots\!, i_1); i_{\ell}'\!\in\!\mathbb Z_{r+1}\}.
\end{align} That
is, $L^{(\ell)}_\alpha$ consists of $\alpha$ as well as the points
in $\mathbb Z_{r+1}^m$ which differs from $\alpha$ only at
the $\ell$th coordinate $($from the right$)$. 

\begin{lem}\label{rp-set}
For each $\alpha\in\Omega_t^{(m)}$ and $\ell\in[m]$, if
$L^{(\ell)}_\alpha\subseteq\Omega_t^{(m)}$, then
$R=L^{(\ell)}_\alpha\backslash\{\alpha\}$ is a recovering set of
$\alpha$.
\end{lem}
\begin{proof}
Let
$$\alpha=(i_m,\cdots,i_{\ell+1},i_\ell,i_{\ell-1},\cdots,i_1).$$
Then by \eqref{mL-alf},
$$L^{(\ell)}_\alpha=\{\alpha_0,\alpha_1,\cdots,\alpha_r\},$$ where
$\alpha_j=(i_m,\cdots,i_{\ell+1},j,i_{\ell-1},\cdots,i_1)$ for
each $j\in\mathbb Z_{r+1}$ and $\alpha=\alpha_{i_\ell}$.

From \eqref{L-alf}, it is easy to see that
\begin{align}\label{mult-layer-set-eq2}
\mathcal L^{(m)}(\alpha_r)=\bigcup_{j=0}^{r-1}\mathcal
L^{(m)}(\alpha_j)\end{align} and for distinct $j_1, j_2\in\mathbb
Z_r$,
\begin{align}\label{mult-layer-set-eq3}
\mathcal L^{(m)}(\alpha_{j_1})\cap \mathcal
L^{(m)}(\alpha_{j_1})=\emptyset.\end{align} So combining
\eqref{symbl-rltn}, \eqref{mult-layer-set-eq2} and
\eqref{mult-layer-set-eq3}, we have
\begin{align*}
x_{\alpha_r}&=\sum_{\beta\in\mathcal
L^{(m)}(\alpha_r)}x_{\beta}\\&=\sum_{j=0}^{r-1}\left(\sum_{\beta\in\mathcal
L^{(m)}(\alpha_j)}x_{\beta}\right)\\&=\sum_{j=0}^{r-1}x_{\alpha_j}\end{align*}
which is equivalent to $($noticing that $\mathcal C$ is a binary
code$)$
$$x_\alpha=\sum_{\beta\in
L^{(\ell)}_\alpha\backslash\{\alpha\}}x_\beta.$$ Note that from
\eqref{mL-alf}, $L^{(\ell)}_\alpha$ has size $r+1$. So
$R=L^{(\ell)}_\alpha\backslash\{\alpha\}$ has size $r$, hence is a
recovering set of $\alpha$.
\end{proof}

\begin{lem}\label{not-in-k}
For any nonempty
$E\subseteq\Omega_{t}^{(m)}\backslash\Omega_{0}^{(m)}$, there
exists an $\alpha\in E$ which has a recovering set
$R\subseteq\Omega_{t}^{(m)}\backslash E$.
\end{lem}
\begin{proof}
Let $s$ be the smallest number such that
$E\cap\Gamma_{s}^{(m)}\neq\emptyset$. Since
$E\subseteq\Omega_{t}^{(m)}\backslash\Omega_{0}^{(m)}$, then
$s\geq 1$ and $\text{supp}_m(s)\neq\emptyset$. Hence, we can
always find a $\ell\in\text{supp}_m(s)$ and a $s'<s$ such that
\begin{align}\label{not-in-k-eq1}
\text{supp}_m(s)=\text{supp}_m(s')\cup\{\ell\}.\end{align}

Pick $\alpha\in E\cap\Gamma_{s}^{(m)}$. Then by \eqref{def-GM},
$\text{\textbf{U}}^{(m)}(\alpha)=\text{supp}_m(s).$ Further, by
\eqref{mL-alf} and \eqref{not-in-k-eq1},
$\text{\textbf{U}}^{(m)}(\beta)=\text{supp}_m(s')$ for each
$\beta\in L^{(\ell)}_\alpha\backslash\{\alpha\}$. Then again by
\eqref{def-GM}, we have
\begin{align}\label{not-in-k-eq2}
L^{(\ell)}_\alpha\backslash\{\alpha\}\subseteq\Gamma_{s'}^{(m)}.
\end{align}
Since $s'<s$ and $s$ is the smallest number such that
$E\cap\Gamma_{s}^{(m)}\neq\emptyset$, then
$E\cap\Gamma_{s'}^{(m)}=\emptyset$. Hence,
$$L^{(\ell)}_\alpha\backslash\{\alpha\}\subseteq
\Gamma_{s'}^{(m)}\backslash E\subseteq\Omega_{s'}^{(m)} \backslash
E\subseteq\Omega_{t}^{(m)} \backslash E,$$ and by Lemma
\ref{rp-set}, $R=L^{(\ell)}_\alpha\backslash\{\alpha\}$ is a
recovering set of $\alpha$.
\end{proof}

Now, we can prove Theorem \ref{main-th}.
\begin{proof}[Proof of Theorem \ref{main-th}]
By the construction, it is easy to see that the code length of
$\mathcal C_t^{(m)}$ is
\begin{align*}
n&=\left|\Omega_t^{(m)}\right|\nonumber\\
&=\sum_{s=0}^t\left|\Gamma_s^{(m)}\right|\nonumber\\
&=\sum_{s=0}^tr^{m-|\text{supp}_m(s)|}\nonumber\\
&=r^m\sum_{s=0}^t\frac{1}{r^{|\text{supp}_m(s)|}},
\end{align*}
and the dimension of $\mathcal C_t^{(m)}$ is
$$k=\left|\Omega_0^{(m)}\right|=r^m.$$
So the code rate is
\begin{align*}
\frac{k}{n}=\frac{1}{\sum_{s=0}^t\frac{1}{r^{|\text{supp}_m(s)|}}}.
\end{align*}

We then need to prove that $\mathcal C_t^{(m)}$ is a $(r,t)$-SLRC.
It is sufficient to prove that for any binary linear code
$\mathcal C$, if $\mathcal C$ has a parity check matrix containing
all rows of $H_t^{(m)}$, then $\mathcal C$ is an $(r,t)$-SLRC. We
will prove this by induction on $m$.

First, for $m=1$, since $1\leq t\leq 2^m-1$, we have $t=1$. By
\eqref{def-GM} and \eqref{def-OM}, $\Gamma_0^{(1)}=\mathbb Z_r$,
$\Gamma_1^{(1)}=\{r\}$ and $\Omega_1^{(1)}=\mathbb Z_{r+1}$. So
$$H_1^{(1)}=(1,1,\cdots,1)_{1\times (r+1)}.$$
Clearly, the binary linear code $\mathcal C$ with parity check
matrix containing $H_1^{(1)}$ is a $(r,1)$-SLRC.

Now, suppose $m>1$ and the induction assumption holds for all
$m'<m$ and $t'\leq 2^{m'}-1$. We consider $m$ and $t\leq 2^{m}-1$.
Using the same notations as in Remark \ref{prtn-n}, we have the
following four claims.
\begin{description}
  \item[i)] For any nonempty $E\subseteq A$ of size
$|E|\leq t_1$, there exists an $\alpha\in E$ such that $\alpha$
has a recovering set $R\subseteq A\backslash E$.
  \item[ii)] For any nonempty $E\subseteq A$ of size
$|E|\leq t$, there exists an $\alpha\in E$ such that $\alpha$ has
a recovering set $R\subseteq\Omega_t^{(m)}\backslash E$.
  \item[iii)] For any nonempty $E\subseteq B$ of size
$|E|\leq t_2$, there exists an $\alpha\in E$ such that $\alpha$
has a recovering set $R\subseteq B\backslash E$.
  \item[iv)] For any nonempty $E\subseteq B$ of size
$|E|\leq t$, there exists an $\alpha\in E$ such that $\alpha$ has
a recovering set $R\subseteq\Omega_t^{(m)}\backslash E$.
\end{description}

We will prove them one by one as follows.

i): Since $E\!\subseteq\!A$ and, by Remark
\ref{prtn-n}, $A\!=\!\bigcup_{j=0}^{r-1}A_j$, then
$E\!\cap\!A_{j_0}\!\neq\!\emptyset$ for some $j_0\!\in\!\mathbb
Z_r$. By 1) of Lemma \ref{Parity-proj}, $\mathcal C|_{A_{j_0}}$
has a parity check matrix containing all rows of
$H_{t_1}^{(m-1)}$. So by induction assumption, $\mathcal
C|_{A_{j_0}}$ is an $(r,t_1)$-SLRC. Moreover, since
$|E\!\cap\!A_{j_0}|\!\leq\!|E|\!\leq\!t_1$, hence, by Lemma
\ref{lem-ELRC}, there exists an $\alpha\!\in\!E\!\cap\!A_{j_0}$
such that $\alpha$ has a recovering set
$R\!\subseteq\!A_{j_0}\!\backslash E\!\subseteq\!A\backslash E$.

ii): According to Remark \ref{prtn-n}, $\{A_j;
j\in\mathbb Z_{r}\}$ is a partition of $A$. We can consider the
following two cases.

Case 1: There are $j_1, j_2\in\mathbb Z_r$, $j_1\neq j_2$, such
that $E\cap A_{j_1}\neq\emptyset$ and $E\cap
A_{j_2}\neq\emptyset$. According to Remark \ref{prtn-n}, $t\leq
2^{m_0}-1=2t_1+1$. Then either $|E\cap A_{j_1}|\leq t_1$ or
$|E\cap A_{j_2}|\leq t_1$. Without loss of generality, assume
$|E\cap A_{j_1}|\leq t_1$. Similar to the proof of 1), $\mathcal
C|_{A_{j_1}}$ is an $(r,t_1)$-SLRC and there exists an $\alpha\in
E\cap A_{j_1}$ such that $\alpha$ has a recovering set $R\subseteq
A_{j_1}\backslash E\subseteq A\backslash E$.

Case 2: $E\subseteq A_{j_1}$ for some $j_1\in\mathbb Z_r$. In this
case, if $E\cap\Omega_{0}^{(m)}=\emptyset$. Then the expected
$\alpha$ exists by Lemma \ref{not-in-k}. So we assume
$E\cap\Omega_{0}^{(m)}\neq\emptyset$. Pick an $\alpha\in
E\cap\Omega_{0}^{(m)}$. Recall that $t_1=2^{m_0-1}-1$. By
\eqref{mL-alf}, we can check that
$L_{\alpha}^{(m_0)}\subseteq\Omega_{0}^{(m)}\bigcup\Gamma_{t_1+1}^{(m)}$
and $L_{\alpha}^{(m_0)}\cap A_{j_1}=\{\alpha\}$. Since $E\subseteq
A_{j_1}$, then $R=L_{\alpha}^{(m_0)}\backslash\{\alpha\}\subseteq
\Omega_{t_1+1}^{(m)}\backslash
E\subseteq\Omega_{t}^{(m)}\backslash E$. By Lemma \ref{rp-set},
$R$ is a recovering set of $\alpha$.

iii): If $t_2=0$, the claim is naturally true.
Assume $t_2\geq 1$. By 2) of Lemma \ref{Parity-proj}, $\mathcal
C|_{B}$ has a parity check matrix containing all rows of
$H_{t_2}^{(m-1)}$. So by induction assumption, $\mathcal C|_{B}$
is an $(r,t_2)$-SLRC. Hence, by Lemma \ref{lem-ELRC}, there exists
an $\alpha\in E$ such that $\alpha$ has a recovering set
$R\subseteq B\backslash E$.

iv): In this case, by the definition of $B$, we
have $E\cap\Omega_{0}^{(m)}=\emptyset$. Hence, by Lemma
\ref{not-in-k}, there exists an $\alpha\in E$ such that $\alpha$
has a recovering set $R\subseteq\Omega_t^{(m)}\backslash E$.

Combining i)-iv) and by Lemma \ref{lem-Com-LRC}, the result
follows.
\end{proof}

\section{Construction from Resolvable Configurations}
In \cite{Balaji16}, by using $t\!-\!3$ mutually orthogonal latin
squares (MOLS) of order $r$, the authors construct a family of
binary $(r,t)$-SLRC with $k=r^2$ and code rate
$\frac{k}{n}\!=\!1\!/\!\left(1+\frac{t-1}{r}+\frac{1}{r^2}\right)$
for odd $t$. A limitation of this construction is $t\leq r+2$,
since, a necessary condition of existing $\ell$ MOLS of order $r
~(r\!>\!1)$ is $\ell\!\leq\!r\!-\!1$ \cite{Colbourn}. In this
section, by using the resolvable configurations, we give a new
family of binary $(r,t)$-SLRC achieving the same rate
$\frac{k}{n}\!=\!1\!/\!\left(1+\frac{t-1}{r}+\frac{1}{r^2}\right)$
for any $r$ and any odd $t\geq 3~($not limited by $t\leq r+2)$.
First, we introduce a definition \cite{Colbourn,Pisanski}.
\begin{defn}\label{Cfg}
Let $X$ be a set of $k$ elements, called points, and $\mathcal A$
be a collection of subsets of $X$, called lines. The pair
$(X,\mathcal A)$ is called a $(k_{t-1}, b_r)$ \emph{configuration}
if the following three conditions hold.
\begin{itemize}
 \item [(1)] Each line contains $r$ points;
 \item [(2)] Each point belongs to $t\!-\!1$ lines;
 \item [(3)] Every pair of distinct points belong to at most
 one line;
\end{itemize}
Clearly, condition (3) is equivalent to the following condition.
\begin{itemize}
 \item [(3$'$)] Every pair of distinct lines have at most
 one point in common;
\end{itemize}
The configuration $(X,\mathcal A)$ is called \emph{resolvable}, if
further
\begin{itemize}
 \item [(4)] All lines in $\mathcal A$ can be partitioned into
 $t\!-\!1$ parallel classes, where a parallel class is a set of lines
 that partition $X$.
\end{itemize}
\end{defn}

For any $(k_{t-1}, b_r)$ resolvable configuration $(X,\mathcal
A)$, one can see that $r|k$ and each parallel class contains
$s=\frac{k}{r}$ lines. So, $b=\frac{k}{r}(t-1)=s(t-1)$ in such a
case. As usual, the incidence matrix of a $(k_{t-1}, b_r)$
configuration $(X,\mathcal A)$, where $X=\{x_1,\cdots,x_k\}$ and
$\mathcal A=\{A_1,\cdots\!,A_b\}$, is defined as a $b\times k$
binary matrix $M=(m_{i,j})$ such that
\begin{equation*}
m_{i,j}=\left\{\begin{aligned}
&1, ~ ~ \text{if}~x_j\in A_i;\\
&0, ~ ~ \text{otherwise}.\\
\end{aligned} \right.
\end{equation*}

Clearly, any configuration is uniquely determined by its incidence
matrix.

\begin{exam}\label{ex-rsvl-cfg}
We can check that the following matrix determines a $(k_{t-1},
b_r)$ resolvable configuration $(X,\mathcal A)$ with $k=9$,
$t-1=4$, $b=12$ and $r=3$. Clearly, $\{A_1,A_2,A_3\}$,
$\{A_4,A_5,A_6\}$, $\{A_7,A_8,A_9\}$ and
$\{A_{10},A_{11},A_{12}\}$ are four parallel classes of
$(X,\mathcal A)$ and any pair of lines in different parallel
classes have one point in common.
\begin{eqnarray*}
M=\left(\begin{array}{ccccccccc}
1 & 1 & 1 & 0 & 0 & 0 & 0 & 0 & 0 \\
0 & 0 & 0 & 1 & 1 & 1 & 0 & 0 & 0 \\
0 & 0 & 0 & 0 & 0 & 0 & 1 & 1 & 1 \\
1 & 0 & 0 & 1 & 0 & 0 & 1 & 0 & 0 \\
0 & 1 & 0 & 0 & 1 & 0 & 0 & 1 & 0 \\
0 & 0 & 1 & 0 & 0 & 1 & 0 & 0 & 1 \\
1 & 0 & 0 & 0 & 0 & 1 & 0 & 1 & 0 \\
0 & 1 & 0 & 1 & 0 & 0 & 0 & 0 & 1 \\
0 & 0 & 1 & 0 & 1 & 0 & 1 & 0 & 0 \\
1 & 0 & 0 & 0 & 1 & 0 & 0 & 0 & 1 \\
0 & 1 & 0 & 0 & 0 & 1 & 1 & 0 & 0 \\
0 & 0 & 1 & 1 & 0 & 0 & 0 & 1 & 0 \\
\end{array}\right)
\end{eqnarray*}
\end{exam}

Resolvable configurations was recently used for constructing codes
whose information symbols have locality $r$ and availability $t$
by Su \cite{Su16}. The author also constructed some resolvable
configurations in the paper, for example, the $(k_{t-1}, b_r)$
resolvable configurations with $k=r^m$ and
$t-1\leq\frac{r^{m-1}}{r-1}$, where $m\geq 2$ and $r$ is a prime
power. The following construction, using the free $\mathbb
Z_r$-module \cite{Brown}, not only generalize the result of
\cite{Su16}, but also enable us to construct $(k_{t-1}, b_r)$
resolvable configuration {\em for any $r,t\geq 2$} ($r$ need not
be a prime power), and further $(r, t)$-SLRCs for any $r\geq 2$
and odd integer $t\geq 3$.

\begin{lem}\label{rscfg-g-r}
For any $r, t\geq 2$, there exists a $(k_{t-1}, b_r)$ resolvable
configuration with $k=r^m$, where $m$ is an arbitrary integer such
that $m\geq \log_2t$.
\end{lem}
\begin{proof}
Consider the free $\mathbb Z_r$-module $X=\mathbb Z_r^m$,
where $\mathbb Z_r$ is the ring of integers modulo $r$. For any
$\alpha\in\mathbb Z_r^m$, we use $\alpha(j)$ to denote the $j$th
coordinate of $\alpha$. For example, if
$\alpha=(i_1,i_{2},\cdots,i_m)$, then $\alpha(j)=i_j$.

For each nonempty $S\subseteq[m]$, let $\alpha_S\in\mathbb Z_r^m$
be such that $\alpha_S(j)=1$ for $j\in S$ and $\alpha_S(j)=0$
otherwise. Let $$A_{S,0}\triangleq\{i\cdot\alpha_S; i\in\mathbb
Z_r\}.$$ Clearly, $A_{S,0}$ is a submodule of $\mathbb Z_r^m$ with
$r$ elements and $A_{S,0}\cap A_{S',0}=(0,0,\cdots,0)$ for any two
distinct nonempty subsets $S$ and $S'$ of $[m]$. Let
$$\mathcal A_S=\{A_{S,\ell},\; \ell=0,1,\cdots,r^{m-1}-1\}$$ be
the collection of all cosets of $A_{S,0}$. Then
$\alpha_1-\alpha_2\in A_{S,0}$ for any
$\ell\in\{0,1,\cdots,r^{m-1}-1\}$ and any $\alpha_1,\alpha_2\in
A_{S,\ell}$.

Note that $m\geq \log_2t$ (i.e., $t-1\leq 2^m-1$) and $[m]$ has
$2^m-1$ nonempty subsets. We can always pick $t-1$ nonempty
subsets of $[m]$, say $S_1,S_2,\cdots, S_{t-1}$. Let
$$\mathcal A=\bigcup_{i=1}^{t-1}\mathcal A_{S_i}.$$
We claim that $(X=\mathbb Z_r^m,\mathcal A)$ is a $(k_{t-1}, b_r)$
resolvable configuration, which can be seen as follows.

Firstly, noticing that for each nonempty $S\subseteq[m]$,
$\mathcal A_S$ is a partition of $X$, then Conditions (1), (2),
(4) of Definition \ref{Cfg} hold. Secondly, if $S$, $S'$ are two
distinct nonempty subsets of $[m]$ and
$\ell,\ell'\in\{0,1,\cdots,2^m-1\}$, then we have $|A_{S,\ell}\cap
A_{S',\ell'}|\leq 1$. Since, if otherwise, suppose $\alpha_1,
\alpha_2\in A_{S,\ell}\cap A_{S',\ell'}$, then
$\alpha_1-\alpha_2\in A_{S,0}\cap A_{S',0}=(0,0,\cdots,0)$. Hence,
we have $\alpha_1=\alpha_2$, i.e. $|A_{S,\ell}\cap
A_{S',\ell'}|\leq 1$. Moreover, since for each nonempty
$S\subseteq[m]$, $\mathcal A_S$ is a partition of $X$, so
Condition (3) of Definition \ref{Cfg} holds, which completes the
proof.
\end{proof}

In the rest of this section, we always assume that $(X,\mathcal
A)$ is a $(k_{t-1}, b_r)$ resolvable configuration and $\mathcal
A\!=\!\{A_1,\cdots\!,A_b\}$. Firstly, we need a lemma for the
property of the resolvable configuration $(X,\mathcal A)$ with odd $t$.

\begin{lem}\label{E-not}
Let $E$ be a $t$-subset of $X$ and $t$ be an odd integer. Then
there exists an $A_j\in \mathcal A$ such that $|E\cap A_j|=1$.
\end{lem}
\begin{proof}
Consider a parallel class of $(X,\mathcal A)$. Since it is a
partition of $X$ and $|E|=t$ is odd, there exists some $A_{j_1}$
in the class such that $|E\cap A_{j_1}|$ is odd. If $|E\cap
A_{j_1}|\!=\!1$, then we have done. So suppose
$E\!=\!\{i_1,\cdots\!,i_t\}$ and
$\{i_1,i_2,i_3\}\!\subseteq\!E\cap A_{j_1}$. Since each point
belongs to $t-1$ lines, we can assume $i_1$ belongs to lines
$A_{j_1},A_{j_2},\cdots,A_{j_{t-1}}$, where
$A_{j_1},A_{j_2},\cdots,A_{j_{t-1}}$ belong to different parallel
classes. Moreover, since every pair of distinct points belong to
at most one line, then $i_2,i_3\notin A_{j_\ell}, \forall
\ell\!\in\!\{2,\cdots,t-1\}$ and each point
$i_\ell,\ell\!\in\!\{4,\cdots\!,t\},$ belongs to at most one line
in $\{A_{j_2},\cdots\!,A_{j_{t-1}}\}$. Hence, there exists a line
$A_{j}\in\{A_{j_2},\cdots\!,A_{j_{t-1}}\}$ that contains no point
in $\{i_2,\cdots\!,i_t\}$. That is to say, $E\cap
A_j\!=\!\{i_1\}$, which completes the proof.
\end{proof}

From now on, we let $X\!=\![k]$ and $\mathcal
A_1\!=\!\{A_1,\cdots\!,A_s\}$, $\mathcal
A_2\!=\!\{A_{s+1},\cdots\!,A_{2s}\}$, $\!\cdots\!$, $\mathcal
A_{t-1}\!=\!\{A_{(t-2)s+1},\cdots\!,A_b\}$ be the $t\!-\!1$
parallel classes of $(X,\mathcal A)$. We further partition $[s]$
into $\lceil\frac{s}{r}\rceil$ nonempty subsets, say
$B_1,\cdots\!,B_{\lceil\frac{s}{r}\rceil}$, such that $|B_i|\leq
r$ for all $i\in\{1,\cdots\!,\lceil\frac{s}{r}\rceil\}$. Such a
partition plays a subtle role in our construction, as will become
clear later. Now, let $W=(w_{i,j})$ be a
$\lceil\frac{s}{r}\rceil\times b$ matrix defined by
\begin{equation*}
w_{i,j}=\left\{\begin{aligned}
&1, ~ ~ \text{if}~j\in B_i\\
&0, ~ ~ \text{otherwise.}\\
\end{aligned} \right.
\end{equation*}
Let $M$ be the incidence matrix of $(X,\mathcal A)$ and
\begin{eqnarray}\label{pcm-cfg}
H=\left(\begin{array}{ccc}
M & I_b & O_{b\times \lceil\frac{s}{r}\rceil}\\
O_{\lceil\frac{s}{r}\rceil\times k} & W & I_{\lceil\frac{s}{r}\rceil}\\
\end{array}\right)
\end{eqnarray}
where $I_\ell$ denotes the $\ell\times \ell$ identity matrix and
$O_{\ell\times\ell'}$ denotes the $\ell\times\ell'$ all-zero
matrix for any positive integers $\ell$ and $\ell'$. Clearly, $H$
has $b+\lceil\frac{s}{r}\rceil$ rows,
$n=k+b+\lceil\frac{s}{r}\rceil$ columns and rank
$b+\lceil\frac{s}{r}\rceil$.

As an example, consider the resolvable configuration $(X,\mathcal
A)$ in Example \ref{ex-rsvl-cfg}. We have $s=\frac{k}{r}=3$ and
$\lceil\frac{s}{r}\rceil=1$. So we can construct
$W=(1,1,1,0,\cdots,0)_{1\times 12}$ and
$I_{\lceil\frac{s}{r}\rceil}=(1)_{1\times 1}$, and according to
\eqref{pcm-cfg}, further construct a matrix $H$ as in
\eqref{H-ex-cd}.
\begin{eqnarray}\label{H-ex-cd}
M=\left(\begin{array}{cccccccccccccccccccccc}
1 & 1 & 1 & 0 & 0 & 0 & 0 & 0 & 0 & 1 & 0 & 0 & 0 & 0 & 0 & 0 & 0 & 0
& 0 & 0 & 0 & 0 \\
0 & 0 & 0 & 1 & 1 & 1 & 0 & 0 & 0 & 0 & 1 & 0 & 0 & 0 & 0 & 0 & 0 & 0
& 0 & 0 & 0 & 0 \\
0 & 0 & 0 & 0 & 0 & 0 & 1 & 1 & 1 & 0 & 0 & 1 & 0 & 0 & 0 & 0 & 0 & 0
& 0 & 0 & 0 & 0 \\
1 & 0 & 0 & 1 & 0 & 0 & 1 & 0 & 0 & 0 & 0 & 0 & 1 & 0 & 0 & 0 & 0 & 0
& 0 & 0 & 0 & 0 \\
0 & 1 & 0 & 0 & 1 & 0 & 0 & 1 & 0 & 0 & 0 & 0 & 0 & 1 & 0 & 0 & 0 & 0
& 0 & 0 & 0 & 0 \\
0 & 0 & 1 & 0 & 0 & 1 & 0 & 0 & 1 & 0 & 0 & 0 & 0 & 0 & 1 & 0 & 0 & 0
& 0 & 0 & 0 & 0 \\
1 & 0 & 0 & 0 & 0 & 1 & 0 & 1 & 0 & 0 & 0 & 0 & 0 & 0 & 0 & 1 & 0 & 0
& 0 & 0 & 0 & 0 \\
0 & 1 & 0 & 1 & 0 & 0 & 0 & 0 & 1 & 0 & 0 & 0 & 0 & 0 & 0 & 0 & 1 & 0
& 0 & 0 & 0 & 0 \\
0 & 0 & 1 & 0 & 1 & 0 & 1 & 0 & 0 & 0 & 0 & 0 & 0 & 0 & 0 & 0 & 0 & 1
& 0 & 0 & 0 & 0 \\
1 & 0 & 0 & 0 & 1 & 0 & 0 & 0 & 1 & 0 & 0 & 0 & 0 & 0 & 0 & 0 & 0 & 0
& 1 & 0 & 0 & 0 \\
0 & 1 & 0 & 0 & 0 & 1 & 1 & 0 & 0 & 0 & 0 & 0 & 0 & 0 & 0 & 0 & 0 & 0
& 0 & 1 & 0 & 0 \\
0 & 0 & 1 & 1 & 0 & 0 & 0 & 1 & 0 & 0 & 0 & 0 & 0 & 0 & 0 & 0 & 0 & 0
& 0 & 0 & 1 & 0 \\
0 & 0 & 0 & 0 & 0 & 0 & 0 & 0 & 0 & 1 & 1 & 1 & 0 & 0 & 0 & 0 & 0 & 0
& 0 & 0 & 0 & 1 \\
\end{array}\right)
\end{eqnarray}
Let $\mathcal C$ be a binary linear code with parity check matrix
$H$ as \eqref{H-ex-cd}. Then from the first $12$ rows of $H$, we
can see that the coordinate $1$ has $4$ disjoint recovering sets,
i.e., $\{2,3,10\}$, $\{4,7,13\}$, $\{6,8,16\}$ and $\{5,9,19\}$,
and the coordinate $10$ has a recovering set
$\{1,2,3\}\subseteq\{1,\cdots,9\}$. Moreover, from the last row of
$H$, we can see that $\{11,12,22\}$ is a recovering set of $10$
and $\{10,11,12\}$ is a recovering set of $22$. In general, we
have the following lemma.
\begin{lem}\label{rem-code-rps}
Let $\mathcal C$ be an $[n,k]$ binary linear code with parity
check matrix $H$ as in \eqref{H-ex-cd}. Then, the following hold.
\begin{itemize}
 \item [1)] Each $i\!\in\![k]$ has $t\!-\!1$ disjoint recovering sets,
 i.e., $A_{j_\ell}\!\cup\!\{k\!+\!j_\ell\}
 \backslash\{i\}$, where $A_{j_\ell}$, $\ell\!=1,\!\cdots\!,t-\!1$,
 are lines containing $i$.
 \item [2)] Each $i\!\in\!\{k\!+\!1,\cdots\!,k\!+\!b\}$ has a recovering set
 $R\!\subseteq\![k]$.
 \item [3)] Each $i\!\in\!\{k\!+\!1,\cdots\!,k\!+\!s\}$ has a recovering set
 $R\!\subseteq\!\{k\!+\!1$, $\cdots$, $k\!+\!s\}\!\cup\!\{k\!+\!b\!+\!1,
 \cdots\!,n\}\backslash\{i\}$.
 \item [4)] Each $i\!\in\!\{k\!+\!b\!+\!1,\cdots\!,n\}$ has a recovering
 set $R\!\subseteq\!\{k\!+\!1$, $\cdots$, $k\!+\!s\}$.
\end{itemize}
\end{lem}
\begin{proof}
1) and 2) are obtained by considering the first $b$ rows of $H$;
3) and 4) are obtained by considering the last
$\lceil\frac{s}{r}\rceil$ rows of $H$.
\end{proof}

\begin{thm}\label{H-be-ELRC}
If $t$ is odd, then the binary linear code $\mathcal C$ with
parity check matrix $H$ as in \eqref{pcm-cfg} is an
$(n,k,r,t)$-SLRC with rate \begin{align*}
\frac{k}{n}=\left(1+\frac{t-1}{r}+\left\lceil\frac{1}{r^2}
\right\rceil\right)^{-1}.\end{align*}
\end{thm}

\begin{proof}
By the construction, $\mathcal C$ has block length
\begin{align*}
n&=k+b+\left\lceil\frac{s}{r}\right\rceil\nonumber\\
&=k\left(1+\frac{t-1}{r}+\left\lceil\frac{1}{r^2}\right\rceil\right).
\end{align*} and dimension
$n-\left(b+\left\lceil\frac{s}{r}\right\rceil\right)=k$.
So the code rate is \begin{align*}
\frac{k}{n}=\left(1+\frac{t-1}{r}+\left\lceil\frac{1}{r^2}
\right\rceil\right)^{-1}.\end{align*}

We now prove, according to Lemma \ref{lem-ELRC}, that for any
$E\subseteq[n]$ with $|E|\leq t$, there exists an $i\in E$ such
that $i$ has a recovering set $R\subseteq[n]\backslash E$.
Consider the following cases.

Case 1: $E\cap[k]=\emptyset$. Then we have
$E\subseteq\{k+1,\cdots\!,n\}$. If
$E\cap\{k+1,\cdots\!,k+b\}\neq\emptyset$, then by 2) of Lemma
\ref{rem-code-rps}, each $i\in E\cap\{k+1,\cdots\!,k+b\}$ has a
recovering set $R\subseteq[k]\subseteq[n]\backslash E$; Otherwise,
$E\subseteq\{k+b+1,\cdots\!,n\}$, then by 4) of Lemma
\ref{rem-code-rps}, each $i\in E$ has a recovering set
$R\subseteq\{k+1,\cdots\!,k+s\}\subseteq\{k+1,\cdots\!,k+b\}
\subseteq[n]\backslash E$.

Case 2: $E\cap[k]\neq\emptyset$. Pick an $i_1\in E\cap[k]$. Let
\begin{align}\label{Def-R-ell}
R_\ell\!=\!A_{j_\ell}\cup\{k+j_\ell\}\backslash\{i_1\},~
\ell\!=\!1,\cdots\!,t\!-\!1,\end{align} where
$A_{j_1}\cdots\!,A_{j_{t-1}}$ are the $t\!-\!1$ lines containing
$i_1$. By 1) of Lemma \ref{rem-code-rps}, $R_1,\cdots,R_{t-1}$ are
$t\!-\!1$ disjoint recovering sets of $i_1$. If
$R_\ell\subseteq[n]\backslash E$ for some
$\ell\!\in\!\{1,\cdots\!,t\!-\!1\}$, then we are done. So we
assume $E\cap R_\ell\neq\emptyset$ for all
$\ell\!\in\!\{1,\cdots\!,t\!-\!1\}$. Since all $R_\ell$s are
disjoint, so $|E|=t$, $|E\cap R_\ell|=1$,
$\ell=1,\cdots\!,t\!-\!1$, and
$E\subseteq\{i_1\}\bigcup\left(\bigcup_{\ell=1}^{t-1}R_\ell\right)$.
We have the following three subcases:

Case 2.1: $E\cap R_\ell\subseteq[k], \forall
\ell\!\in\!\{1,\cdots\!,t\!-\!1\}$. Then $E\subseteq[k]$. Since
$|E|=t$ is odd, by Lemma \ref{E-not}, $|E\cap A_i|=1$ for some
$A_i\in\mathcal A$. Let $E\cap A_i=\{i_2\}$. By 1) of Lemma
\ref{rem-code-rps},
$R=A_i\cup\{k+i\}\backslash\{i_2\}\subseteq[n]\backslash E$ is a
recovering set of $i_2$.

Case 2.2: $E\cap R_{\ell_1}\subseteq[k]$ and $E\cap
R_{\ell_2}\nsubseteq[k]$ for some
$\{\ell_1,\ell_2\}\subseteq\{1,\cdots\!,t\!-\!1\}$. Without loss
of generality, assume $i_2\in E\cap R_{1}\subseteq[k]$ and $E\cap
R_{2}\nsubseteq[k]$. Then according to \eqref{Def-R-ell}, we have
$E\cap R_{2}=\{k+j_2\}$. By 1) of Lemma \ref{rem-code-rps}, we can
let $R_1',\cdots\!,R_{t-1}'$ are $t\!-\!1$ disjoint recovering
sets of $i_2$, where $R_1'=A_{j_1}\cup\{k+j_1\}\backslash\{i_2\}$
and $R_\ell'=A_{j_\ell'}\cup\{k+j_\ell'\}\backslash\{i_2\}$,
$\ell=2,\cdots\!,t\!-\!1$, such that $A_{j_1}$ together with
$A_{j_2'},\cdots\!,A_{j_{t-1}'}$ are the $t-1$ lines containing
$i_2~($see Fig. \ref{fig-cnt-rp}$)$. Note that $A_{j_1}$ is the
only line containing both $i_1$ and $i_2$, one can see that
$\{i_1,i_2,k+j_2\}\cap R_\ell'=\emptyset, \forall
\ell\in\{2,\cdots\!,t-1\}$. Since $|E|=t$, then there exists an
$\ell_0\in\{2,\cdots\!,t-1\}$ such that $E\cap
R_{\ell_0}'=\emptyset$. Hence, $R_{\ell_0}'\subseteq[n]\backslash
E$ is a recovering sets of $i_2$.

Case 2.3: $E\cap R_{\ell}\nsubseteq[k]$ for all
$\ell\in\{1,\cdots\!,t\!-\!1\}$. Then we have
$E\!\cap\!R_\ell\!=\!\{k\!+\!j_\ell\}$. Note that
$A_{j_1},\cdots\!, A_{j_{t-1}}$ belong to distinct parallel
classes (since all of them contain $i_1$). Without loss of
generality, we assume $A_{j_\ell}\!\in\!\mathcal A_\ell$,
$\ell\!\in\!\{1,\cdots\!,t\!-\!1\}$. Then $j_{1}\leq s$ and
$s<j_\ell\leq b$, $\ell=2,\cdots,t-1$. By 3) of Lemma
\ref{rem-code-rps}, $k\!+\!j_{1}$ has a recovering set
$R\!\subseteq\!\{k\!+\!1$, $\cdots$,
$k\!+\!s\}\!\cup\!\{k\!+\!b\!+\!1,\cdots\!,n\}\backslash\{j_1\}
\!\subseteq\![n]\backslash E$.

By the above discussion, for any $E\subseteq[n]$ of size $|E|\leq
t$, there exists an $i\in E$ such that $i$ has a recovering set
$R\subseteq[n]\backslash E$. So by Lemma \ref{lem-ELRC}, $\mathcal
C$ is an $(n,k,r,t)$-SLRC.
\end{proof}
\renewcommand\figurename{Fig}
\begin{figure}[htbp]
\begin{center}
\includegraphics[height=5cm]{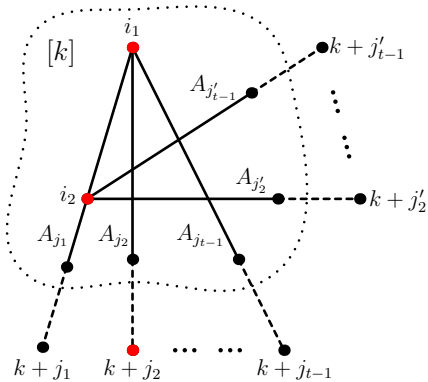}
\end{center}
\caption{Illustration of points and recovering sets:
$R_\ell=A_{j_\ell}\cup\{k+j_\ell\}\backslash\{i_1\}$,
$\ell=1,\cdots,t-1$, are $t-1$ recovering sets of $i_1$; $R_1$ and
$R'_\ell=A_{j_\ell'}\cup\{k+j_\ell'\}\backslash\{i_2\}$,
$\ell=2,\cdots,t-1$, are $t-1$ recovering sets of $i_2$.}
\label{fig-cnt-rp}
\end{figure}

\section{Conclusions and Future Work}
In this paper, we investigated sequential locally repairable codes
(SLRC) by proposing an upper bound on the code rate of $(n, k, r,
t)$-SLRC for $t=3$, and constructed two families of $(n, k, r,
t)$-SLRC for $r,t\geq 2$ (for the second family, $t$ is odd). Both
of our constructions have code rate $>\frac{r}{r+t}$ and are
optimal for $t\in\{2,3\}$ with respect to the proposed bound.

It is still an open problem to determine the optimal code rate of
$(n, k, r, t)$-SLRCs for general $t$, i.e., $t\geq 5$. Here, we
conjecture that an achievable upper bound of the code rate of $(n,
k, r, t)$-SLRCs has the following form:
\begin{align}\label{gnr-bnd-n}
\frac{k}{n}\leq\left(1+\sum_{i=1}^m\frac{a_i}{r^i}\right)^{-1},\end{align}
where $m=\left\lceil\log_rk\right\rceil$, all $a_i\geq 0$ are
integers such that $\sum_{i=1}^ma_i=t$. This conjecture can be
verified for $t\in\{1,2,3,4\}$, for which the values of the
$m$-tuple $(a_1,\cdots,a_m)$, denoted by $\alpha_t$ for each $t$,
are listed in the following table, where, the cases of $t=2,3$ are
due to \cite{Prakash-14} and this work, respectively. The case of
$t=4$ (for binary code) is recently due to Balaji et al
\cite{Balaji16-2}. \vspace{0.2cm}\begin{center}
\begin{tabular}{|p{1.0cm}|p{3.5cm}|}
\hline \small{$t$} & \small{$a_1 ~~~ a_2 ~~~ a_3 { ~~~} \cdots ~~~ a_m$}\\
\hline \small{$1$} & \small{$~1$ ~~ $0$ ~~ $0$ ~~~ $\cdots$ ~~ $0$}  \\
\hline \small{$2$} & \small{$~2$ ~~ $0$ ~~ $0$ ~~~ $\cdots$ ~~ $0$}  \\
\hline \small{$3$} & \small{$~2$ ~~ $1$ ~~ $0$ ~~~ $\cdots$ ~~ $0$} \\
\hline \small{$4$} & \small{$~2$ ~~ $2$ ~~ $0$ ~~~ $\cdots$ ~~ $0$} \\
\hline
\end{tabular}\\
\vspace{0.15cm}\footnotesize{Table 1. The known values of
$\alpha_t=(a_1,a_2,\cdots,a_m)$.}
\end{center}

It is very hard to give the explicit values of $\alpha_t$ for
general $t\geq 5$, even a LP-based or recursive formulation of
$\alpha_t$ seems difficult. Further, we conjecture that
$\alpha_5=(2,2,1,0,\cdots,0)$ and $\alpha_6=(2,3,1,0,\cdots,0)$.
We would like to take the above problems our future work.

\end{document}